%% file: main.tex
\documentclass[11pt]{article}
\usepackage{amsmath, amssymb, amsthm, graphicx}
\usepackage{amsfonts}
\usepackage{geometry}
\usepackage{hyperref}

\usepackage{listings}
\usepackage{xcolor}
\usepackage{braket}
\usepackage{qcircuit}
\usepackage{booktabs}
\usepackage{physics}
\usepackage{natbib}
\usepackage{algorithm}
\usepackage{algorithmic}

\usepackage{tikz}
\usetikzlibrary{shapes.geometric, arrows, positioning, fit, calc}

\definecolor{gruvbox-bg}{RGB}{40, 40, 40}
\definecolor{gruvbox-fg}{RGB}{235, 219, 178}
\definecolor{gruvbox-red}{RGB}{204, 36, 29}
\definecolor{gruvbox-green}{RGB}{152, 151, 26}
\definecolor{gruvbox-yellow}{RGB}{215, 153, 33}
\definecolor{gruvbox-blue}{RGB}{69, 133, 136}
\definecolor{gruvbox-purple}{RGB}{177, 98, 134}
\definecolor{gruvbox-aqua}{RGB}{104, 157, 106}
\definecolor{gruvbox-orange}{RGB}{214, 93, 14}
\definecolor{gruvbox-gray}{RGB}{168, 153, 132}
\definecolor{gruvbox-light-red}{RGB}{251, 73, 52}
\definecolor{gruvbox-light-green}{RGB}{184, 187, 38}
\definecolor{gruvbox-light-yellow}{RGB}{250, 189, 47}
\definecolor{gruvbox-light-blue}{RGB}{131, 165, 152}
\definecolor{gruvbox-light-purple}{RGB}{211, 134, 155}
\definecolor{gruvbox-light-aqua}{RGB}{142, 192, 124}
\definecolor{gruvbox-light-orange}{RGB}{254, 128, 25}
\usepackage{geometry}

\newsavebox{\lcufigure}
\savebox{\lcufigure}{%
\Qcircuit @C=0.8em @R=0.8em {
    \lstick{|0\rangle} & \gate{H} & \ctrlo{1} & \ctrl{1} & \gate{H} & \meter \\
    \lstick{|0\rangle} & \qw & \gate{U_{+}} & \gate{U_{-}} & \qw & \qw
}
}

\newtheorem{theorem}{Theorem}
\newtheorem*{theorem*}{Theorem}
\newtheorem{definition}{Definition}
\newtheorem{lemma}{Lemma}
\newtheorem{corollary}{Corollary}
\newtheorem{proposition}{Proposition}
\newtheorem{remark}{Remark}

\title{Classical and Quantum Algorithms for Topological Invariants of Torus Bundles}

\author{
    Nelson Abdiel Col\'on Vargas \\
    \textit{University of Cambridge} \\
    \texttt{nac52@cam.ac.uk}
    \and 
    Carlos Ortiz Marrero \\
    \textit{Department of Computer Science, Colorado State University}\\
    \texttt{carlos.ortiz.marrero@colostate.edu}
}
\date{\today}

\begin{document}

\maketitle

\begin{abstract}
Computing topological invariants of 3-manifolds is generally intractable, yet specialized algebraic structures can enable efficient algorithms. For Witten-Reshetikhin-Turaev (WRT) invariants of torus bundles, we exploit the non-commutative torus structure to embed the skein algebra of the closed torus into its symmetric subalgebra at roots of unity. This yields a fixed $N^2$-dimensional representation that supports polynomial-time classical computation with $O(N^2)$ space, and a quantum algorithm using only $O(\log N)$ qubits—an exponential space advantage. We further prove that extracting individual expansion coefficients is \#P-complete, yet there is a quantum algorithm that can efficiently approximate these coefficients for a non-negligible fraction of configurations.
\end{abstract}

\section{Introduction}

The Witten-Reshetikhin-Turaev (WRT) invariants \cite{witten1989,reshetikhin1991} are fundamental tools for distinguishing three-dimensional spaces (3-manifolds), arising from quantum field theory and extending the Jones polynomial \cite{jones1985} from knots to full 3-manifold topology. However, computing these invariants for general 3-manifolds is \#P-hard \cite{alagic2017}.

Despite this general intractability, torus bundles, i.e.\ mapping tori of elements $g \in SL_2(\mathbb{Z})$ acting on the torus, are a notable exception. Their special algebraic structure allows us to embed the Kauffman bracket skein algebra
$K_N(\Sigma_{1,0})$ into the symmetric subalgebra of the non-commutative torus
\cite{kohenfrohman2020,frohmangelca2000}. This yields at roots of unity a fixed $N^2$-dimensional representation where the Frohman-Gelca product rule's exponentially many expansion terms collapse into a manageable set of coefficients. Unlike braid-based quantum algorithms for knot invariants \cite{aharonov2009,laakkonen2025}, which compute properties of 1-dimensional curves within a fixed 3-manifold and require qubits linear in braid width, our approach computes invariants of the 3-manifold itself using only $O(\log N)$ qubits.

We present three main results exploiting the non-commutative torus embedding. First, we develop an explicit classical dynamic programming algorithm with $\Theta((m+\ell)N^2)$ time and $\Theta(N^2)$ space, where $m$ is the number of skein insertions and $\ell$ is the monodromy word length; from the topological quantum computation perspective, this computes amplitudes for $SU(2)_k$ anyons on the torus \cite{brennen2008}. Second, we develop a space-efficient quantum algorithm using $O(\log N + m)$ qubits versus $\Theta(N^2)$ classical memory, implementing skein algebra operations coherently via Linear Combination of Unitaries—an exponential space advantage. Third, we prove that a coefficient extraction subproblem is \#P-complete for exact computation via parsimonious reductions from \#SUBSET-SUM, yet admits polynomial-time quantum additive approximation, demonstrating provable computational separation.

\subsection{The computational problem}
Throughout we fix an odd integer $N\ge 3$, set
\[
t \;=\; e^{2\pi i/N},
\]
and work modulo $N$ on indices in $\mathbb{Z}_N^2$. Phases arise from the symplectic form $\omega((p,s),(r,u)) = pu - sr$.
\begin{definition}[WRT invariant for torus bundles]
\label{def:wrt-trace}
For a torus bundle $M_g$ (mapping torus of $g \in SL_2(\mathbb{Z})$) with skein insertions $x_1,\ldots,x_m \in K_N(\Sigma_{1,0})$, the WRT invariant at level $k = N-2$ is:
\[
Z_N(M_g; x_1,\ldots,x_m) = \operatorname{Tr}(\rho(g) L_{x_m} \cdots L_{x_1})
\]
where $\rho(g)$ is the modular action and $L_{x_i}$ are left-multiplication operators in the skein algebra.
\end{definition}

This computational problem naturally decomposes into two tasks: computing the full trace (the WRT invariant itself) and exactly determining specific coefficients from the operator representation. As we will see, these tasks have fundamentally different computational complexity. While trace computation is polynomial-time classically via dynamic programming on the fixed $N^2$-dimensional coefficient table, exact coefficient extraction becomes \#P-hard when $N$ is large enough to prevent modular aliasing, revealing where quantum advantages truly emerge (quantum algorithms can efficiently approximate these coefficients even when exact computation is intractable).

\subsection{Main results}

\begin{theorem*}[Classical Dynamic Programming Algorithm]
WRT invariants $Z_N(M_g; x_1,\ldots,x_m)$ for torus bundles can be computed classically in $\Theta((m+\ell)N^2)$ time and $\Theta(N^2)$ space using dynamic programming on the coefficient table of the non-commutative torus, where $\ell = |g|_{S,T}$ is the word length of the monodromy.
\end{theorem*}

The algorithm maintains an $N^2$-dimensional coefficient table in the non-commutative torus, updating it incrementally for each skein insertion and monodromy generator—avoiding explicit enumeration of all $2^m$ expansion terms. The word length $\ell$ grows logarithmically in the matrix entries: for $g = \bigl(\begin{smallmatrix} a & b \\ c & d \end{smallmatrix}\bigr) \in SL_2(\mathbb{Z})$, we have $\ell = O(\log \max(|a|,|b|,|c|,|d|))$ via Euclidean algorithm decomposition \cite{farb2011}.

\begin{theorem*}[Quantum Algorithm]
There exists a quantum algorithm that approximates the normalized WRT amplitude $\widetilde{Z} = \frac{1}{N^2}\operatorname{Tr}(W(g; x_1,\ldots,x_m))$ for torus bundles $M_g$ to additive precision $\epsilon$ with probability at least $1-\delta$ using:
\begin{align}
\text{Qubits:} &\quad 2\times\lceil\log_2 N \rceil + m + 1 \text{ ($O(\log N)$ data qubits, $m+1$ ancillas)} \notag\\
\text{Base circuit depth:} &\quad D = O(\ell \log^2 N + m\log^2 N) \text{ where } \ell = |g|_{S,T} \notag\\
\text{Total operations:} &\quad O(D \cdot \log(1/\delta)/\epsilon^2) \text{ via sampling} \notag
\end{align}
This provides a space-efficient alternative to the classical algorithm, using $O(\log N + m)$ versus $O(N^2)$ space.
\end{theorem*}

The quantum algorithm encodes the $N^2$-dimensional coefficient table into $O(\log N)$ qubits and implements skein operations coherently. The exponential space advantage ($O(\log N)$ qubits versus $\Theta(N^2)$ classical memory) becomes significant for large $N$ or memory-constrained quantum hardware.

\begin{theorem*}[Quantum Advantage for Coefficient Counting]
The FG-Coefficient Counting problem (Definition~\ref{def:fgcc}) is \#P-complete under parsimonious reductions from \#SUBSET-SUM when the modulus $N$ exceeds $\sum_i |a_i| + |z|$, preventing modular aliasing. Despite this classical hardness \cite{valiant1979}, the quantum algorithm (Theorem~\ref{thm:fgcc-quantum}) estimates the normalized coefficient $\alpha = 2^{-m}c(z)$ to additive precision $\epsilon$ with probability at least $1-\delta$ using $O(\log N + m)$ total qubits, base circuit depth $D = O(m \log^2 N)$, and $O(D \cdot \log(1/\delta)/\epsilon^2)$ total operations via sampling. This additive approximation is meaningful for coefficients of size $\Omega(2^m)$.
\end{theorem*}

This establishes a computational separation: extracting coefficient $c(z)$ requires counting solutions to a sum—as hard as any problem in \#P—while quantum superposition enables polynomial-time additive approximation via the Hadamard test. These results rest on the following structural foundation:

\begin{theorem*}[Structural Decomposition]
The embedding $K_N(\Sigma_{1,0}) \hookrightarrow \mathcal{W}_t^{\iota}$ (symmetric subalgebra of the non-commutative torus) yields a uniform decomposition where each skein generator admits representation as $L_{B_{(p,s)}} = U_{p,s} + U_{-p,-s}$ with $U_{p,s}$ phased-permutation unitaries implementable in depth $O(\log^2 N)$.
\end{theorem*}

This uniform two-term structure (Theorems~\ref{thm:fg-embedding} and \ref{thm:lcu-decomp})—in contrast to variable-width decompositions in localized skein algebras—enables both efficient classical dynamic programming and quantum LCU implementation.

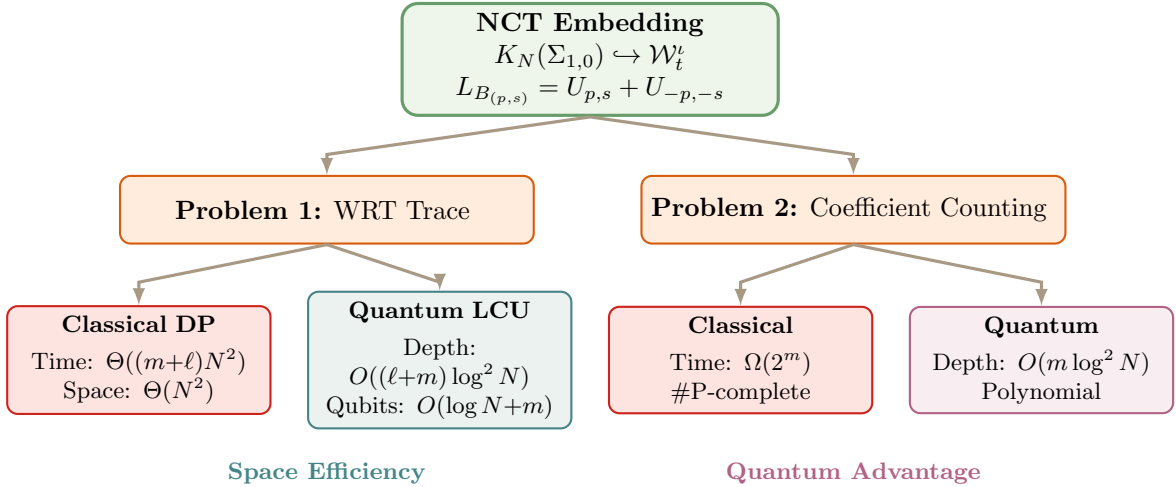
\begin{figure}[t]
\centering
\begin{tikzpicture}[
    node distance=0.8cm and 1.5cm,
    foundation_box/.style={
        rectangle,
        rounded corners=6pt,
        draw=gruvbox-aqua,
        very thick,
        fill=gruvbox-light-aqua!15,
        minimum width=5cm,
        minimum height=1.2cm,
        align=center,
        font=\small
    },
    problem_box/.style={
        rectangle,
        rounded corners=4pt,
        draw=gruvbox-orange,
        thick,
        fill=gruvbox-light-orange!15,
        minimum width=5.5cm,
        minimum height=0.9cm,
        align=center,
        font=\small
    },
    approach_box/.style={
        rectangle,
        rounded corners=4pt,
        draw=gruvbox-gray,
        thick,
        minimum width=3.5cm,
        minimum height=1.3cm,
        align=center,
        text width=3.2cm,
        font=\footnotesize
    },
    arrow_style/.style={
        -latex,
        very thick,
        draw=gruvbox-gray
    }
]

\node[foundation_box] (foundation) at (0,0) {
    \textbf{NCT Embedding} \\
    $K_N(\Sigma_{1,0}) \hookrightarrow \mathcal{W}_t^{\iota}$ \\
    $L_{B_{(p,s)}} = U_{p,s} + U_{-p,-s}$
};

\node[problem_box] (trace_problem) at (-3.5, -2) {
    \textbf{Problem 1:} WRT Trace
};

\node[problem_box] (coeff_problem) at (3.5, -2) {
    \textbf{Problem 2:} Coefficient Counting
};

\draw[arrow_style] (foundation.south) -- ++(-3.5, -0.5) -- (trace_problem.north);
\draw[arrow_style] (foundation.south) -- ++(3.5, -0.5) -- (coeff_problem.north);

\node[approach_box, fill=gruvbox-light-red!15, draw=gruvbox-red] (classical_trace) at (-6, -4) {
    \textbf{Classical DP} \\[3pt]
    Time: $\Theta((m{+}\ell)N^2)$ \\
    Space: $\Theta(N^2)$
};

\node[approach_box, fill=gruvbox-light-blue!15, draw=gruvbox-blue] (quantum_trace) at (-2, -4) {
    \textbf{Quantum LCU} \\[3pt]
    Depth: $O((\ell{+}m) \log^2 N)$ \\
    Qubits: $O(\log N {+} m)$
};

\draw[arrow_style] (trace_problem.south) -- ++(-2.5, -0.5) -- (classical_trace.north);
\draw[arrow_style] (trace_problem.south) -- ++(1.5, -0.5) -- (quantum_trace.north);

\node[approach_box, fill=gruvbox-light-red!15, draw=gruvbox-red] (classical_coeff) at (2, -4) {
    \textbf{Classical} \\[3pt]
    Time: $\Omega(2^m)$ \\
    \#P-complete
};

\node[approach_box, fill=gruvbox-light-purple!15, draw=gruvbox-purple] (quantum_coeff) at (6, -4) {
    \textbf{Quantum} \\[3pt]
    Depth: $O(m \log^2 N)$ \\
    Polynomial
};

\draw[arrow_style] (coeff_problem.south) -- ++(-1.5, -0.5) -- (classical_coeff.north);
\draw[arrow_style] (coeff_problem.south) -- ++(2.5, -0.5) -- (quantum_coeff.north);

\node[font=\footnotesize, text=gruvbox-blue, align=center] at (-3.5, -5.5) {\textbf{Space Efficiency}};
\node[font=\footnotesize, text=gruvbox-purple, align=center] at (3.5, -5.5) {\textbf{Quantum Advantage}};

\end{tikzpicture}
\caption{Computational landscape enabled by the NCT embedding. The structural foundation (top) enables two distinct computational problems. \textbf{Left (WRT Trace):} The classical dynamic programming algorithm achieves $\Theta((m{+}\ell)N^2)$ time but requires $\Theta(N^2)$ space. The quantum LCU approach maintains polynomial time $O((\ell{+}m) \log^2 N)$ using $O(\log N + m)$ total qubits compared to $\Theta(N^2)$ classical memory—an exponential space advantage. \textbf{Right (Coefficient Counting):} The classical enumeration is \#P-complete, requiring exponential $\Omega(2^m)$ time and thus intractable for exact counting. The quantum Hadamard test achieves polynomial-time additive approximation with $O(m \log^2 N)$ circuit depth, demonstrating proven computational advantage via coherent superposition over all $2^m$ terms.}
\label{fig:three_results}
\end{figure}

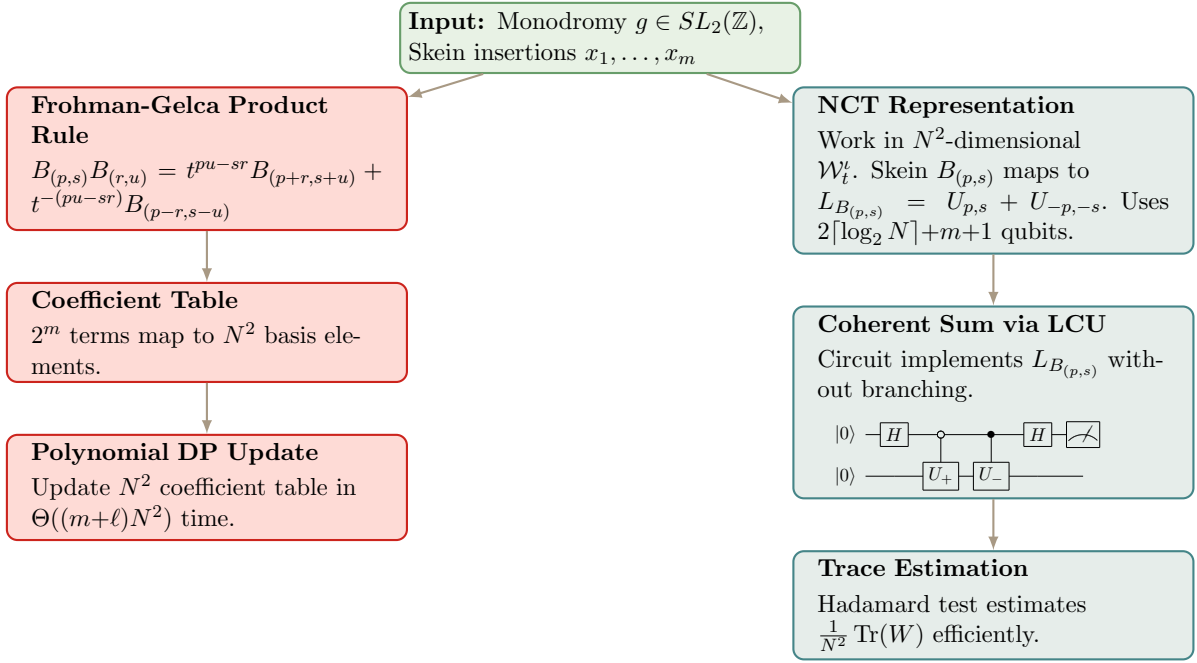
\begin{figure}[t]
\centering
\begin{tikzpicture}[scale=0.85, transform shape,
    node distance=0.8cm and 1.2cm, 
    box/.style={
        rectangle, 
        rounded corners, 
        draw=gruvbox-gray, 
        thick,
        minimum height=1.0cm,
        text centered,
        text width=6cm,
        align=center
    },
    path_arrow/.style={
        -latex,
        thick,
        draw=gruvbox-gray
    },
    classical_box/.style={
        box,
        fill=gruvbox-light-red!20,
        draw=gruvbox-red
    },
    quantum_box/.style={
        box,
        fill=gruvbox-light-blue!20,
        draw=gruvbox-blue
    },
    title/.style={
        font=\bfseries\large
    }
]

\node[title] (classical_title) {Classical Approach};
\node[title, right=8cm of classical_title] (quantum_title) {Quantum Approach};

\coordinate (midpoint_titles) at ($(classical_title)!0.5!(quantum_title)$);
\node[box, fill=gruvbox-light-aqua!20, draw=gruvbox-aqua, below=1cm of midpoint_titles] (input_box) {
    \parbox{6cm}{\textbf{Input:} Monodromy $g \in SL_2(\mathbb{Z})$, \\ Skein insertions $x_1, \ldots, x_m$}
};

\node[classical_box, below=2cm of classical_title] (classical_multiply) {
    \parbox{5.5cm}{
        \textbf{Frohman-Gelca Product Rule} \\[2pt]
        $B_{(p,s)} B_{(r,u)} = t^{pu-sr}B_{(p+r,s+u)} + t^{-(pu-sr)}B_{(p-r,s-u)}$
    }
};
\draw[path_arrow] (input_box) -- (classical_multiply);

\node[classical_box, below=of classical_multiply] (classical_expand) {
    \parbox{5.5cm}{
        \textbf{Coefficient Table} \\[2pt]
        $2^m$ terms map to $N^2$ basis elements.
    }
};
\draw[path_arrow] (classical_multiply) -- (classical_expand);

\node[classical_box, below=of classical_expand] (classical_result) {
    \parbox{5.5cm}{
        \textbf{Polynomial DP Update} \\[2pt]
        Update $N^2$ coefficient table in $\Theta((m{+}\ell)N^2)$ time.
    }
};
\draw[path_arrow] (classical_expand) -- (classical_result);

\node[quantum_box, below=2cm of quantum_title] (quantum_embed) {
    \parbox{5.5cm}{
        \textbf{NCT Representation} \\[2pt]
        Work in $N^2$-dimensional $\mathcal{W}_t^{\iota}$. Skein $B_{(p,s)}$ maps to $L_{B_{(p,s)}} = U_{p,s} + U_{-p,-s}$. Uses $2\lceil\log_2 N \rceil {+} m {+} 1$ qubits.
    }
};
\draw[path_arrow] (input_box) -- (quantum_embed);

\node[quantum_box, below=of quantum_embed] (quantum_lcu) {
    \parbox{5.5cm}{
        \textbf{Coherent Sum via LCU} \\[2pt]
        Circuit implements $L_{B_{(p,s)}}$ without branching. \\[3pt]
        \scalebox{0.75}{\hspace{1cm}\usebox{\lcufigure}}
    }
};
\draw[path_arrow] (quantum_embed) -- (quantum_lcu);

\node[quantum_box, below=of quantum_lcu] (quantum_result) {
    \parbox{5.5cm}{
        \textbf{Trace Estimation} \\[2pt]
        Hadamard test estimates $\frac{1}{N^2}\operatorname{Tr}(W)$ efficiently.
    }
};
\draw[path_arrow] (quantum_lcu) -- (quantum_result);

\end{tikzpicture}
\caption{Classical vs quantum approaches to WRT computation. \textbf{Classical (left):} The Frohman–Gelca two-term rule generates $2^m$ terms conceptually, but these map to an $N^2$ coefficient table enabling polynomial-time dynamic programming with $\Theta(N^2)$ space. \textbf{Quantum (right):} By embedding the skein algebra into the symmetric subalgebra $\mathcal{W}_t^{\iota}$ of the non-commutative torus, we implement each multiplication as a coherent sum via LCU, using $O(\log N)$ data qubits plus $m+1$ ancillas—an exponential space reduction from the $\Theta(N^2)$ classical memory requirement.}
\label{fig:classical_vs_quantum}
\end{figure}

\subsection{Comparison with existing approaches}

Previous quantum algorithms for topological invariants, pioneered by the AJL algorithm \cite{aharonov2009}, have primarily focused on the Jones polynomial via braid group representations. While the AJL algorithm efficiently approximates Jones polynomials for links in $S^3$, our method targets WRT invariants of skein-decorated torus bundles. Recent implementations \cite{laakkonen2025} have demonstrated braid-based approaches on quantum hardware, establishing experimental benchmarks for quantum algorithms in topological quantum field theory.

The technical approach underlying braid-based Jones polynomial algorithms differs fundamentally from our non-commutative torus framework. Standard braid-based methods require choosing a place (central character) that affects the resulting matrix representation. This localization introduces two technical challenges: first, each skein element's Pauli decomposition width depends on both the element itself and the chosen place, leading to variable ancilla requirements; second, the localized structure restricts computation to link invariants rather than general 3-manifold topology. In contrast, our embedding into the symmetric subalgebra $\mathcal{W}_t^{\iota}$ of the non-commutative torus yields a fixed $N^2$-dimensional representation independent of any place choice (detailed comparison in Appendix~\ref{app:place-free}). This place-free structure provides uniform two-term decompositions $L_{B_{(p,s)}} = U_{p,s} + U_{-p,-s}$ for all skein elements, enabling quantum implementation via Linear Combination of Unitaries (with one ancilla per skein multiplication for correct block-encoding composition) and direct computation of trace-type WRT amplitudes for general mapping tori.

Beyond these technical differences, a fundamental distinction lies in the dimensionality of the mathematical objects being characterized. Braid-based quantum algorithms \cite{aharonov2009,laakkonen2025} compute knot invariants—topological properties of 1-dimensional curves (knots and links) embedded within a fixed 3-manifold (typically $S^3$). Our algorithm computes 3-manifold invariants, characterizing topological properties of the 3-dimensional space $M_g$ itself. Both approaches arise from the same TQFT framework \cite{witten1989,reshetikhin1991}, but they apply to fundamentally different levels of structure: knot invariants describe objects \emph{within} a 3-manifold, while our WRT invariants describe the 3-manifold itself. The Jones polynomial is recovered as a special case when the 3-manifold is $S^3$ and the WRT invariant includes link insertions \cite{turaev1994}.

This dimensional distinction manifests in resource scaling. Braid-based algorithms require $n$ qubits to compute the Jones polynomial of an $(n-1)$-strand braid closure, exhibiting linear qubit scaling in the braid width. Our algorithm requires $2\lceil\log_2 N\rceil + m + 1$ qubits to compute WRT invariants in the level-$N$ representation with $m$ skein insertions, demonstrating logarithmic scaling in the modular parameter (the $m+1$ ancillas scale with the number of insertions, not the representation dimension). For instance, with 14 data qubits, braid-based methods handle 13-strand braids while our approach accommodates $N=127$ (dimension 16,129). This exponential advantage in qubit efficiency reflects the different geometric structures: modular arithmetic on $\mathbb{Z}_N^2$ versus Fibonacci representations of braid groups.

Beyond braid-based approaches to knot invariants and our skein algebra embedding for torus bundles, alternative computational frameworks for 3-manifold invariants include quantum Teichm{\"u}ller theory \cite{chekhov1999,andersenkashaev2014}, which computes WRT-type invariants via state integrals using Faddeev's quantum dilogarithm, and triangulation-based methods implemented in computational topology software such as Regina \cite{burton2013} and SnapPy \cite{snappy}, which compute Turaev-Viro invariants for general 3-manifolds at the cost of exponential worst-case complexity. These approaches address different scopes and mathematical frameworks: quantum Teichm{\"u}ller theory provides a complementary perspective grounded in hyperbolic geometry and volume conjectures, while triangulation methods achieve greater generality (arbitrary 3-manifolds) at the expense of computational efficiency. Our focus on torus bundles exploits the specialized structure where the skein algebra embeds into the non-commutative torus, which at roots of unity becomes finite-dimensional, enabling polynomial-time computation unavailable for general 3-manifolds.

\subsection{Paper organization}

Section 2 establishes the mathematical foundations: Witten-Reshetikhin-Turaev invariants, mapping tori, and the Kauffman bracket skein algebra. Section 3 presents the key mathematical framework—the embedding of the skein algebra into the symmetric subalgebra of the non-commutative torus $\mathcal{W}_t^{\iota}$—and develops the classical dynamic programming algorithm, achieving $\Theta((m+\ell)N^2)$ time and $\Theta(N^2)$ space by exploiting the fixed $N^2$-dimensional structure this embedding provides. Section 4 develops the quantum algorithm, deriving the two-term operator decomposition and implementing it via linear combinations of unitaries to achieve $O(\log N + m)$ total qubit scaling. Section 5 proves that the FG-Coefficient Counting problem is \#P-complete and presents a quantum algorithm with computational advantage for this hard problem. Section 6 concludes with implications for quantum algorithm design and topological quantum computation. Three appendices provide additional context: Appendix~\ref{app:place-free} explains the place-free representation approach, Appendix~\ref{app:mapping} provides background on mapping tori and modular actions, and Appendix~\ref{app:technical-anyons} describes the precise correspondence between our algorithms and $SU(2)_k$ anyon models.

We begin by establishing the mathematical foundations underlying our computational problem: the topological objects (WRT invariants, mapping tori) and algebraic structures (skein algebra) that we seek to compute.

\section{Mathematical Foundations}
\label{sec:foundations}

We establish three mathematical foundations: Witten-Reshetikhin-Turaev invariants, the class of torus bundles (mapping tori) we consider, and the Kauffman bracket skein algebra. Although computing WRT invariants of general 3-manifolds is \#P-hard \cite{alagic2017}, torus bundles admit efficient computation due to their special algebraic structure.

\subsection{Witten-Reshetikhin-Turaev invariants}

Topological quantum field theory (TQFT) provides a mathematical framework for extracting topological invariants of manifolds from quantum-mechanical structures. In 1989, Witten showed that the Jones polynomial of a knot could be understood via three-dimensional Chern-Simons gauge theory with gauge group SU(2) \cite{witten1989}. This path-integral approach yielded not only a new perspective on the Jones polynomial but also invariants of arbitrary closed 3-manifolds.

Reshetikhin and Turaev \cite{reshetikhin1991} provided a rigorous mathematical construction of these invariants using quantum group representation theory, bypassing the analytical difficulties of the path integral. Their construction produces what are now called Witten-Reshetikhin-Turaev (WRT) invariants—topological invariants of 3-manifolds parameterized by a level $k$ (or equivalently, by $N = k + 2$ where $t = e^{2\pi i/N}$ is a root of unity).

\begin{definition}[Witten-Reshetikhin-Turaev TQFT]
A $(2+1)$-dimensional WRT theory at level $k$ assigns:
\begin{align}
\text{To each surface } \Sigma: &\quad \text{a finite-dimensional Hilbert space } V_N(\Sigma) \label{eq:tqft-surface} \\
\text{To each cobordism } M: &\quad \text{a linear map } Z(M): V_N(\partial_\text{in} M) \to V_N(\partial_\text{out} M) \label{eq:tqft-cobordism} \\
\text{To each closed 3-manifold } M: &\quad \text{a complex number } Z_N(M) \in \mathbb{C} \label{eq:tqft-closed}
\end{align}
These assignments satisfy the TQFT axioms \cite{atiyah1988}: functoriality under composition of cobordisms, multiplicativity under disjoint union, and normalization (the empty manifold maps to $\mathbb{C}$). Independence of decomposition choices follows from these axioms \cite{turaev1994}.
\end{definition}

For the torus $\Sigma_{1,0} = T^2$, the standard SU(2) WRT Hilbert space $V_N(\Sigma_{1,0})$ has dimension $N-1$ (corresponding to integrable representations at level $k = N-2$). In this work, however, we operate in the $N^2$-dimensional Weyl (non-commutative torus) representation to compute trace-type amplitudes via normalized traces. This $N^2$ representation admits an action of the \emph{mapping class group} of the torus—the group of isotopy classes of orientation-preserving diffeomorphisms—which is isomorphic to $SL_2(\mathbb{Z})$ \cite{farb2011}. Appendix~\ref{app:place-free} explains why this place-free representation is advantageous for quantum computation.

This TQFT structure specializes to mapping tori as follows. The cylinder $T^2 \times [0,1]$ is a cobordism from $T^2$ to itself; to construct the mapping torus $M_g$, we glue the boundary components via the monodromy $g \in SL_2(\mathbb{Z})$. The TQFT assigns to the $g$-twisted cylinder a linear map $\rho(g): V_N(T^2) \to V_N(T^2)$, called the \emph{modular representation}. When we close this cobordism by identifying the incoming and outgoing boundaries, assignment~\eqref{eq:tqft-closed} gives the WRT invariant $Z_N(M_g)$, which by functoriality equals the trace $\operatorname{Tr}(\rho(g))$. With skein elements $x_1, \ldots, x_m$ inserted along the fiber direction, this generalizes to
\[
Z_N(M_g; x_1,\ldots,x_m) = \operatorname{Tr}(\rho(g) L_{x_m} \cdots L_{x_1}),
\]
the trace formula of Definition~\ref{def:wrt-trace}.

\subsection{Mapping tori}

A mapping torus is a torus bundle over $S^1$ constructed from a monodromy $g \in SL_2(\mathbb{Z})$. These are the 3-manifolds whose WRT invariants we compute.

\begin{definition}[Mapping torus]
For $g \in SL_2(\mathbb{Z})$, the \emph{mapping torus} is the 3-manifold
\[
M_g = \frac{T^2 \times [0,1]}{(x,1) \sim (g \cdot x,0)},
\]
where $g$ acts on $T^2 = \mathbb{R}^2/\mathbb{Z}^2$ by the standard linear action. The 3-manifold $M_g$ is obtained by gluing the boundary tori of the cylinder $T^2 \times [0,1]$ via the diffeomorphism induced by $g$, called the \emph{monodromy} \cite{hatcher3manifolds}.
\end{definition}

\subsection{The Kauffman bracket skein algebra}

The skein algebra provides the algebraic setting for the multiplication operators $L_{x_i}$ in the trace formula (Definition~\ref{def:wrt-trace}). For a surface $\Sigma$, the Kauffman bracket skein algebra $K_N(\Sigma)$ is generated by isotopy classes of framed links (links equipped with a continuous choice of normal vector, or equivalently, embedded annuli), with multiplication given by stacking and relations encoding the Kauffman bracket.

\begin{definition}[Kauffman bracket skein algebra]
\label{def:skein-algebra}
The \emph{Kauffman bracket skein algebra} $K_N(\Sigma)$ is the $\mathbb{C}$-module generated by isotopy classes of framed links in $\Sigma \times [0,1]$, modulo the skein relations:
\begin{align}
\langle L_+ \rangle &= A \langle L_0 \rangle + A^{-1} \langle L_\infty \rangle \label{eq:skein-crossing} \\
\langle L \cup \bigcirc \rangle &= (-A^2 - A^{-2}) \langle L \rangle \label{eq:skein-loop}
\end{align}
where $A = e^{\pi i/N}$, and $L_+, L_0, L_\infty$ denote three links differing only at a single crossing (with $L_0, L_\infty$ the two resolutions). Multiplication is given by stacking in the $[0,1]$ direction.
\end{definition}

\begin{remark}[Parameter conventions]
\label{rem:skein-parameters}
Throughout the paper we use two notational styles for the Kauffman bracket skein algebra. For a general nonzero parameter we follow the common convention and write $K_t(\Sigma)$, where $t\in\mathbb{C}^{*}$ is the skein parameter (and in the quantum-torus model one often writes $t=e^{2\pi i\theta}$). When specializing to a root-of-unity \emph{level}, we write $K_N(\Sigma)$ for the specialization
where the Kauffman bracket parameter $A$ is a primitive $2N$-th root of unity with $N$ odd.
In this setting it is convenient to also record the squared parameter
\[
t := A^2,
\]
so that $t$ is a primitive $N$-th root of unity; equivalently, if $t=e^{2\pi i\theta}$ then $\theta=\frac{1}{N}$.
\end{remark}

For the torus $\Sigma_{1,0}$, elements of $K_N(\Sigma_{1,0})$ correspond to linear combinations of curves on the torus. Each non-contractible simple closed curve is characterized up to isotopy by its slope $(p,s) \in \mathbb{Z}^2 \setminus \{(0,0)\}$, representing $p$ windings around the meridian and $s$ around the longitude; the pairs $(p,s)$ and $(-p,-s)$ represent the same curve with opposite orientations.

\begin{definition}[Primitive pairs]
\label{def:primitive}
A pair $(p,s) \in \mathbb{Z}^2$ is \emph{primitive} if $\gcd(p,s) = 1$. Primitive pairs correspond to simple (non-self-intersecting) closed curves on the torus.
\end{definition}


\begin{definition}[Chebyshev polynomials]
\label{def:chebyshev}
The Chebyshev polynomials (in trace normalization) are defined by $T_0 = 2$, $T_1 = x$, and $T_{n+1} = x \cdot T_n - T_{n-1}$. 
\end{definition}

\begin{definition}[Threaded notation {\cite{frohmangelca2000}}]
\label{def:threaded}
For $(p,s) \in \mathbb{Z}^2$ with $d = \gcd(p,s)$, the \emph{threaded element} is
\[
(p,s)_T := T_d\bigl((p/d, s/d)\bigr),
\]
where $\gamma = (p/d, s/d)$ is the primitive curve and $T_d(\gamma)$ denotes the Chebyshev polynomial evaluated at the skein element $\gamma$ using skein algebra multiplication: $T_0(\gamma) = 2$, $T_1(\gamma) = \gamma$, and $T_{d+1}(\gamma) = \gamma \cdot T_d(\gamma) - T_{d-1}(\gamma)$.
\end{definition}

The key algebraic structure is the Frohman--Gelca product rule, which expresses the product of any two threaded elements as a sum of exactly two terms:

\begin{proposition}[Two-term multiplication {\cite{frohmangelca2000}}]
\label{prop:fg-product}
In $K_t(\Sigma_{1,0})$ with $t \in \mathbb{C}^{*}$, for any $(p,s), (r,u) \in \mathbb{Z}^2$:
\[
(p,s)_T \cdot (r,u)_T = t^{pu-sr}(p+r,s+u)_T + t^{-(pu-sr)}(p-r,s-u)_T.
\]
\end{proposition}

This two-term structure is essential for efficient computation: it bounds the growth of terms when multiplying skein elements, enabling the polynomial-time algorithms developed in Sections~\ref{sec:math-framework} and~\ref{sec:quantum}.

\begin{remark}
    For a detailed account of the construction of $K_N(F)$ we refer the reader to our earlier work with Frohman \cite{abdielfrohman2016}.
\end{remark}

\subsection{Complexity landscape}
\label{sec:complexity-landscape}

The \#P-hardness of WRT invariants for general 3-manifolds \cite{alagic2017} and the Jones polynomial for arbitrary links \cite{jaeger1990} raises a natural question: what structural property makes torus bundles tractable?

The answer lies in the algebraic structure established above. Three properties combine to enable efficient computation:
\begin{enumerate}
\item \textbf{Two-term multiplication.} The Frohman--Gelca product rule (Proposition~\ref{prop:fg-product}) generates exactly two terms per multiplication, bounding the growth of intermediate expressions.
\item \textbf{Finite-dimensional embedding.} At an $N$-th root of unity, the skein algebra embeds into a fixed $N^2$-dimensional space---formalized in Section~\ref{sec:math-framework} via the Frohman--Gelca embedding theorem.
\item \textbf{Modular arithmetic.} Indices reduce modulo $N$, preventing unbounded growth and ensuring all computations remain within the finite basis.
\end{enumerate}
These structural properties transform an apparently exponential problem---$2^m$ terms from $m$ multiplications---into one amenable to polynomial-time dynamic programming and efficient quantum circuits, as developed in Sections~\ref{sec:math-framework} and~\ref{sec:quantum}.

\section{Mathematical Framework: The Symmetric Subalgebra Structure}
\label{sec:math-framework}

Having established the topological and algebraic objects (Section~\ref{sec:foundations}), we now present the mathematical framework that enables efficient computation: the embedding of the skein algebra into the symmetric subalgebra of the non-commutative torus. The key challenge in computing WRT invariants via skein algebra multiplication is that the Frohman-Gelca product rule generates $2^m$ terms—exponential in the number of insertions $m$. However, these terms are not independent: they all belong to the $N^2$-dimensional skein algebra basis. The computational question is whether this exponential expansion can be organized efficiently within the fixed-dimensional structure.

The answer lies in the embedding of the skein algebra $K_N(\Sigma_{1,0})$ into the symmetric subalgebra of the non-commutative torus $\mathcal{W}_t$. This embedding, originally established by Frohman and Gelca \cite{frohmangelca2000} in the context of character varieties within the WRT TQFT framework \cite{reshetikhin1991}, provides exactly the algebraic structure needed for efficient computation. Specifically, it yields:
\begin{enumerate}
\item A fixed $N^2$-dimensional representation (enabling classical table storage)
\item Uniform two-term operator decompositions (enabling quantum LCU implementation)
\item Single-term multiplication in the ambient algebra (enabling efficient updates)
\end{enumerate}

This section establishes the embedding theorem (Theorem~\ref{thm:fg-embedding}) and its computational implications for both classical and quantum algorithms. Throughout, we emphasize the ``place-free'' nature of this representation—it requires no choice of central character, yielding universal computational structures independent of localization choices. Appendix~\ref{app:place-free} provides detailed comparison with place-based approaches, and Appendix~\ref{app:technical-anyons} describes the physical interpretation in terms of $SU(2)_k$ anyons.

\subsection{The non-commutative torus}

The key structure is the non-commutative torus algebra, where multiplication has a particularly simple form amenable to both classical and quantum computation.

\begin{definition}[Non-commutative torus \cite{rieffel1981}]
The non-commutative torus algebra $\mathcal{W}_t$ is the universal $C^*$-algebra generated by unitaries $l, m$ satisfying:
\[
m \cdot l = t^2 \cdot l \cdot m
\]
where $t$ is a primitive $N$-th root of unity. Throughout this paper, we take $t = e^{2\pi i/N}$, consistent with Remark~\ref{rem:skein-parameters}.
\end{definition}

The algebra $\mathcal{W}_t$ is also called the \emph{quantum torus} at parameter $t$. It is a deformation of the commutative $C^*$-algebra $C(T^2)$ of continuous functions on the 2-torus: when $t = 1$, the relation $ml = t^2 lm$ becomes $ml = lm$, recovering $C(T^2)$ with dense subalgebra $\mathbb{C}[l^{\pm 1}, m^{\pm 1}]$ (Laurent polynomials). For $t = e^{2\pi i\theta}$ with $\theta$ irrational, $\mathcal{W}_t$ is the irrational rotation algebra $A_\theta$, a fundamental object in non-commutative geometry \cite{rieffel1981}. Our choice $t = e^{2\pi i/N}$ (rational $\theta = 1/N$) yields a finite-dimensional quotient: imposing $l^N = m^N = 1$ gives $\dim \mathcal{W}_t = N^2$, the computationally tractable setting for level-$N$ WRT invariants.

This algebra has a natural basis that will connect directly to our skein algebra embedding:

\begin{proposition}[Weyl basis]
\label{prop:weyl-mult}
$\mathcal{W}_t$ has basis $\{e_{p,s} = t^{-ps} l^p m^s : p,s \in \mathbb{Z}_N\}$ with multiplication:
\[
e_{p,s} \cdot e_{r,u} = t^{pu-sr} e_{p+r,s+u}
\]
\end{proposition}

The single-term multiplication structure follows from the commutation relations of the Weyl operators \cite{rieffel1981}.

\begin{remark}[Single-term multiplication structure]
The key computational advantage of the Weyl basis is that multiplication produces exactly one basis element with a phase factor: $e_{p,s} \cdot e_{r,u} = t^{pu-sr} e_{p+r,s+u}$. This contrasts with the two-term Frohman-Gelca product (Proposition~\ref{prop:fg-product}), which generates two basis elements. The single-term structure in $\mathcal{W}_t$ is what enables efficient classical dynamic programming updates (equation~\eqref{eq:dp-update}).
\end{remark}

\subsection{The symmetric subalgebra and skein algebra embedding}

\begin{definition}[Orientation-reversal involution]
The involution $\iota: \mathcal{W}_t \to \mathcal{W}_t$ defined by $\iota(e_{p,s}) = e_{-p,-s}$ generates the symmetric subalgebra:
\[
\mathcal{W}_t^{\iota} = \{w \in \mathcal{W}_t : \iota(w) = w\}
\]
\end{definition}

The involution $\iota$ captures a fundamental symmetry that, remarkably, corresponds exactly to the algebraic structure of the skein algebra. This leads to the main structural result:

\begin{theorem}[Frohman--Gelca embedding {\cite{frohmangelca2000}}]
\label{thm:fg-embedding}
Let $t\in\mathbb{C}^\times$ and let $\mathcal{W}_t$ be the noncommutative torus with basis $\{e_{p,s}\}_{(p,s)\in\mathbb{Z}^2}$
and product $e_{p,s}\,e_{r,u}=t^{\,pu-sr}\,e_{p+r,s+u}$. Let $\iota:\mathcal{W}_t\to\mathcal{W}_t$ be the involution
$\iota(e_{p,s})=e_{-p,-s}$ and write $\mathcal{W}_t^{\iota}$ for its fixed subalgebra.
There is an algebra isomorphism
\[
\phi_t:\ K_t(\Sigma_{1,0})\xrightarrow{\ \cong\ }\mathcal{W}_t^{\iota}
\]
characterized on threaded elements by
\[
\phi_t\big((p,s)_T\big)\;=\;e_{p,s}+e_{-p,-s}\qquad \text{for all }(p,s)\in\mathbb{Z}^2.
\]
\end{theorem}

The embedding theorem holds for generic $t \in \mathbb{C}^\times$. For algorithmic applications, we specialize to roots of unity, which yields finite-dimensional representations suitable for computation.

\begin{corollary}[Odd level root-of-unity specialization]
\label{cor:odd-level}
Assume $A$ is a primitive $2N$-th root of unity with $N$ odd and set $t=A^2$ (so $t$ has order $N$).
Let $U=e_{1,0}$, $V=e_{0,1}$ and form the standard $N\times N$ ``clock--shift'' quotient/representation $\rho_N: \mathcal{W}_t \to \operatorname{Mat}_N(\mathbb{C})$
by imposing $U^N=V^N=\mathbf{1}$. Then:
\begin{enumerate}
\item The images of $\{(p,s)_T\}$ under $\rho_N\circ\phi_t$ depend only on $(p,s)\bmod N$
(up to an overall normalization factor depending only on $N$), and the involution $(p,s)\mapsto(-p,-s)$ identifies opposite classes.
\item Consequently, the image of $K_t(\Sigma_{1,0})$ is the $\iota$-fixed subalgebra of $\operatorname{Mat}_N(\mathbb{C})$
spanned by $\{\,U^{p}V^{s}+U^{-p}V^{-s}\mid (p,s)\in(\mathbb{Z}/N\mathbb{Z})^2\,\}$, which has dimension
\[
\dim \operatorname{Im}(\rho_N\circ\phi_t)\;=\;\frac{N^2+1}{2}.
\]
\end{enumerate}
\end{corollary}

\begin{proof}
The first claim follows from the periodicity $U^N = V^N = \mathbf{1}$, which implies $U^{p+N} = U^p$ and $V^{s+N} = V^s$.
For the dimension count: the involution $(p,s) \mapsto (-p,-s)$ on $\mathbb{Z}_N^2$ has exactly one fixed point when $N$ is odd, namely $(0,0)$, since $2p \equiv 0 \pmod{N}$ implies $p = 0$ for $N$ odd (and similarly for $s$). The remaining $N^2 - 1$ elements partition into $(N^2-1)/2$ orbits of size 2. Each orbit contributes one basis element $U^pV^s + U^{-p}V^{-s}$ to the symmetric subalgebra, giving total dimension $1 + (N^2-1)/2 = (N^2+1)/2$.
\end{proof}

\begin{remark}[Scope]
Theorem~\ref{thm:fg-embedding} is an isomorphism of infinite-dimensional algebras for generic $t$. The finite-dimensional statement in Corollary~\ref{cor:odd-level} arises from the root-of-unity specialization $t = e^{2\pi i/N}$, which yields a quotient representation; the skein algebra $K_N(\Sigma_{1,0})$ itself remains infinite-dimensional.
\end{remark}

The embedding theorem has immediate algorithmic consequences. By representing skein elements in the Weyl basis of $\mathcal{W}_t^{\iota}$, we can track all $2^m$ terms of the Frohman-Gelca expansion within a fixed $N^2$-dimensional coefficient table. We now present the classical dynamic programming algorithm that exploits this structure.

\subsection{Algorithmic overview}

The key insight is that while the Frohman-Gelca product rule generates $2^m$ terms conceptually, these terms all map to a fixed $N^2$-dimensional space. This dimensional reduction occurs because we work at the root of unity $t = e^{2\pi i/N}$, which induces modular arithmetic: although the skein algebra $K_N(\Sigma_{1,0})$ is originally indexed by primitive pairs $(p,s) \in \mathbb{Z}^2$ (unrestricted integers), the embedding into $\mathcal{W}_t^{\iota}$ causes all indices to reduce modulo $N$. We can thus maintain a coefficient table $C: \mathbb{Z}_N^2 \to \mathbb{C}$ and update it efficiently for each skein multiplication.

For skein elements $x_i = B_{v_i}$ indexed by primitive pairs $v_i = (p_i, s_i) \in \mathbb{Z}^2$ (which reduce to representatives in $\mathbb{Z}_N^2$ at the root of unity), we compute:
\[
P_m = \prod_{i=1}^m L_{B_{v_i}} = \sum_{w \in \mathbb{Z}_N^2} C_m(w) U_w
\]
where $U_w := e_w$ denotes the Weyl basis operator from Proposition~\ref{prop:weyl-mult}.

\subsection{The dynamic programming update}

Starting with $C_0(0,0) = 1$ and $C_0(w \neq (0,0)) = 0$, we perform the following update for each skein element $v = v_{k+1}$:

\begin{equation}\label{eq:dp-update}
\boxed{
C_{k+1}(w) = t^{\langle v, w \rangle} C_k(w - v) + t^{-\langle v, w \rangle} C_k(w + v)
}\quad\text{(all indices modulo $N$)}
\end{equation}

where $\langle (p,s), (r,u) \rangle = pu - sr$ is the symplectic form (recall $t = e^{2\pi i/N}$). For a fixed skein element $v = (p,s)$, we define the \emph{phase grid} $\Phi: \mathbb{Z}_N^2 \to \mathbb{C}$ by $\Phi(w) = t^{\langle v, w \rangle}$, so that the update~\eqref{eq:dp-update} becomes $C_{k+1}(w) = \Phi(w) \cdot C_k(w-v) + \overline{\Phi(w)} \cdot C_k(w+v)$.

\begin{algorithm}
\caption{Classical DP for WRT Invariants of Torus Bundles}
\label{alg:classical-dp}
\begin{algorithmic}[1]
\REQUIRE Primitive pairs $v_1, \ldots, v_m \in \mathbb{Z}^2$ (reduced mod $N$), monodromy $g \in SL_2(\mathbb{Z})$
\ENSURE WRT invariant $Z_N(M_g; B_{v_1}, \ldots, B_{v_m})$
\STATE Initialize $C: \mathbb{Z}_N^2 \to \mathbb{C}$ with $C(0,0) \gets 1$ and $C(w) \gets 0$ otherwise
\FOR{$i = 1$ to $m$}
    \STATE Let $v = v_i = (p, s) \in \mathbb{Z}_N^2$
    \STATE Compute phase grid $\Phi(w) = t^{\langle v, w \rangle}$ for all $w \in \mathbb{Z}_N^2$
    \STATE Form cyclic shifts $C^+(w) = C(w - v)$, $C^-(w) = C(w + v)$ \hfill (indices mod $N$)
    \STATE Update $C(w) \gets \Phi(w) \cdot C^+(w) + \overline{\Phi(w)} \cdot C^-(w)$ for all $w$
\ENDFOR
\STATE Form operator $P = \sum_{w \in \mathbb{Z}_N^2} C(w) \, e_w$ in the Weyl basis
\STATE Compute $Z_N = \operatorname{Tr}\bigl(\rho(g) \cdot P\bigr)$
\RETURN $Z_N(M_g; B_{v_1}, \ldots, B_{v_m})$
\end{algorithmic}
\end{algorithm}

\begin{theorem}[Complexity of Classical Algorithm]
Algorithm \ref{alg:classical-dp} computes WRT invariants for torus bundles with time complexity $\Theta(mN^2)$ for the dynamic programming phase plus $\Theta(\ell N^2)$ for computing the monodromy action plus $\Theta(N^2)$ for the trace, yielding $\Theta((m+\ell)N^2)$ total time and $\Theta(N^2)$ space.
\end{theorem}

\begin{proof}
Each of the $m$ updates applies the Frohman-Gelca product rule (Proposition~\ref{prop:fg-product}), which generates two terms per existing coefficient. The update~\eqref{eq:dp-update} processes all $N^2$ entries of the coefficient table $C: \mathbb{Z}_N^2 \to \mathbb{C}$, performing for each entry: one phase computation $t^{\langle v, w \rangle}$ (constant time via precomputed lookup table or $O(\log N)$ modular exponentiation), two coefficient lookups $C(w \pm v)$ with modular index arithmetic, two complex multiplications, and one complex addition. This gives $\Theta(N^2)$ operations per update and $\Theta(mN^2)$ total for the dynamic programming phase.

The monodromy action $\rho(g)$ is computed by composing $\ell$ generators, where each generator $\rho(S)$ or $\rho(T)$ acts on the $N^2$-dimensional space via structured operations (2D Fourier transform for $S$, diagonal-times-permutation for $T$), each costing $\Theta(N^2)$. This yields $\Theta(\ell N^2)$ for the monodromy phase. The final trace computation evaluates $Z_N = \sum_{w \in \mathbb{Z}_N^2} C_m(w) \cdot \tau_g(w)$, where $\tau_g(w) = \operatorname{Tr}(\rho(g) \cdot e_w)$ is extracted from the precomputed monodromy in $\Theta(N^2)$ total time. Space is dominated by storing the $N^2$ complex coefficients of $C$.
\end{proof}

The classical algorithm achieves $\Theta((m+\ell)N^2)$ time complexity by materializing the full coefficient table. We now turn to the quantum algorithm, which trades space for coherence: instead of storing $N^2$ complex coefficients explicitly, it maintains superpositions in $O(\log N)$ data qubits---logarithmic data qubit scaling while preserving polynomial-time trace estimation.

Before proceeding to the quantum algorithm, we note an alternative physical interpretation of the classical computation that connects to topological quantum computation.

\begin{remark}[TQFT perspective: anyon braiding interpretation]
From the topological quantum field theory viewpoint \cite{atiyah1988,witten1989,reshetikhin1991}, the classical dynamic programming algorithm can be understood as simulating anyon braiding in a $(2+1)$-dimensional TQFT. Each skein element $B_{(p,s)}$ corresponds to creating an anyon pair with quantum numbers determined by $(p,s)$, and left multiplication by $L_{B_{(p,s)}}$ implements the braiding operation. The coefficient table $C: \mathbb{Z}_N^2 \to \mathbb{C}$ tracks the amplitudes of different anyon fusion outcomes. This perspective connects our computational approach to the broader framework of topological quantum computation \cite{brennen2008}, where WRT invariant evaluation is realized through modular functor operations on the torus Hilbert space.
\end{remark}

\section{The Quantum Algorithm}
\label{sec:quantum}

With the symmetric subalgebra structure established (Section~\ref{sec:math-framework}), we now develop our quantum algorithm. The key enabling structure is the uniform two-term decomposition of skein elements, which we first derive, then implement via linear combinations of unitaries (LCU) \cite{childs2012}. While the classical algorithm efficiently manages the $2^m$ terms on an $N^2$ table, our quantum approach maintains coherent superpositions using only $O(\log N)$ data qubits plus $m+1$ ancillas, trading memory for the ability to directly estimate traces without materializing the full coefficient table.

\subsection{Operator representation and two-term decomposition}

To implement skein multiplication in the quantum setting, we translate the algebraic embedding into operator language. The key is to represent skein elements as linear operators acting on the non-commutative torus, which then admit decompositions into quantum-implementable gates.

\begin{definition}[Left multiplication operators]
\label{def:left-mult}
For $x \in \mathcal{W}_t$, define $L_x: \mathcal{W}_t \to \mathcal{W}_t$ by $L_x(y) = x \cdot y$.
\end{definition}

\begin{theorem}[Two-term LCU decomposition]
\label{thm:lcu-decomp}
For each skein generator $B_{(p,s)} \in \mathcal{W}_t^{\iota}$:
\[
L_{B_{(p,s)}} = U_{p,s} + U_{-p,-s}
\]
where $U_{p,s}$ are phased-permutation unitaries:
\[
U_{p,s}: e_{r,u} \mapsto t^{pu-sr} e_{p+r,s+u}
\]
\end{theorem}

\begin{proof}
By the Frohman-Gelca embedding (Theorem~\ref{thm:fg-embedding}), each skein generator maps to $B_{(p,s)} = e_{p,s} + e_{-p,-s} \in \mathcal{W}_t^{\iota}$. For any basis element $e_{r,u} \in \mathcal{W}_t$, we compute the left multiplication action:
\begin{align*}
L_{B_{(p,s)}}(e_{r,u}) &= B_{(p,s)} \cdot e_{r,u} = (e_{p,s} + e_{-p,-s}) \cdot e_{r,u} \\
&= e_{p,s} \cdot e_{r,u} + e_{-p,-s} \cdot e_{r,u} \\
&= t^{pu-sr} e_{p+r,s+u} + t^{(-p)u-(-s)r} e_{-p+r,-s+u} \quad \text{(Proposition~\ref{prop:weyl-mult})} \\
&= U_{p,s}(e_{r,u}) + U_{-p,-s}(e_{r,u})
\end{align*}
where $U_{p,s}: e_{r,u} \mapsto t^{pu-sr} e_{p+r,s+u}$.

To verify that $U_{p,s}$ is unitary: the index map $(r,u) \mapsto (r+p, u+s)$ is a bijection on $\mathbb{Z}_N^2$ (with inverse $(r,u) \mapsto (r-p, u-s)$), and the phase factor satisfies $|t^{pu-sr}| = 1$ since $t = e^{2\pi i/N}$ lies on the unit circle. Thus $U_{p,s}$ permutes the orthonormal basis $\{e_{r,u}\}$ with unit-modulus phases, which characterizes a unitary operator.
\end{proof}

The symmetric subalgebra structure is essential: while the full non-commutative torus has single-term multiplication, only the symmetric subalgebra provides the two-term decomposition matching the skein algebra's Frohman-Gelca rule. This exact correspondence allows us to implement skein multiplication as a coherent quantum sum.

\begin{remark}
Both classical and quantum algorithms exploit the same algebraic structure but in different ways. Classically, we use dynamic programming to update an $N^2$-dimensional coefficient table, efficiently combining the $2^m$ conceptual terms as they map to the same basis elements. In the quantum approach, we maintain coherent superpositions in $O(\log N)$ qubits, trading space for the ability to estimate traces directly without computing the full coefficient table. The quantum approach is particularly advantageous when $N$ is large and we only need the trace, not the full operator representation.
\end{remark}

\subsection{Trace estimation via the Hadamard test}

The two-term decomposition enables LCU implementation of skein multiplication. To extract the WRT invariant---a trace---we use the Hadamard test \cite{cleve1998}, a fundamental quantum subroutine for estimating inner products and traces.

\begin{proposition}[Hadamard test for trace estimation]
\label{prop:hadamard-test}
Given a unitary $U$ acting on an $n$-qubit system prepared in a maximally mixed state $\rho_{\text{mix}} = I/2^n$, the normalized trace $\widetilde{Z} = \operatorname{Tr}(U)/2^n$ can be estimated as follows. Initialize an ancilla qubit in $|0\rangle$, apply $H$ (Hadamard), then controlled-$U$, then $H$ again. The measurement probabilities are:
\begin{align}
P(\text{measure } |0\rangle) &= \frac{1}{2}\left(1 + \text{Re}[\widetilde{Z}]\right), \\
P(\text{measure } |1\rangle) &= \frac{1}{2}\left(1 - \text{Re}[\widetilde{Z}]\right).
\end{align}
Applying a phase gate $S = \text{diag}(1, i)$ before the final Hadamard extracts $\text{Im}[\widetilde{Z}]$ via the same measurement. 
\end{proposition}

\begin{proof}
The initial state is $\rho_0 = |0\rangle\langle 0| \otimes \rho_{\text{mix}}$. After the circuit $(H \otimes I) \circ \text{C-}U \circ (H \otimes I)$, the diagonal components of the final density matrix in the ancilla basis are:
\[
\frac{1}{4}|0\rangle\langle 0| \otimes (I + U)\rho_{\text{mix}}(I + U^\dagger) + \frac{1}{4}|1\rangle\langle 1| \otimes (I - U)\rho_{\text{mix}}(I - U^\dagger).
\]
The probability of measuring $|0\rangle$ on the ancilla is:
\[
P(0) = \frac{1}{4}\operatorname{Tr}[(I + U)\rho_{\text{mix}}(I + U^\dagger)] = \frac{1}{4 \cdot 2^n}\operatorname{Tr}[(I+U)(I+U^\dagger)].
\]
Expanding $(I+U)(I+U^\dagger) = 2I + U + U^\dagger$ and using $\operatorname{Tr}(U^\dagger) = \overline{\operatorname{Tr}(U)}$ yields $P(0) = \frac{1}{2}(1 + \text{Re}[\operatorname{Tr}(U)/2^n])$. For the imaginary part, inserting $S = \text{diag}(1,i)$ before the final Hadamard replaces $(I \pm U)$ with $(I \pm iU)$ in the interference terms, yielding $P(0) = \frac{1}{2}(1 + \text{Im}[\operatorname{Tr}(U)/2^n])$.
\end{proof}

\begin{remark}[Relation to DQC1]
Proposition~\ref{prop:hadamard-test} is precisely the DQC1 (Deterministic Quantum Computation with one clean qubit) protocol \cite{knill1998,shor2008}: one clean ancilla qubit controlling a unitary on $n$ maximally mixed qubits. The maximally mixed input is essential because the Hadamard test with a pure state $|\psi\rangle$ estimates only $\langle\psi|U|\psi\rangle$---a single expectation value---whereas the trace $\operatorname{Tr}(U) = \sum_j \langle j|U|j\rangle$ requires summing over all diagonal elements. The maximally mixed state $\rho_{\text{mix}} = I/2^n = \frac{1}{2^n}\sum_j |j\rangle\langle j|$ achieves this implicitly, yielding $\operatorname{Tr}(U \cdot \rho_{\text{mix}}) = \operatorname{Tr}(U)/2^n$. An equivalent pure-state formulation uses the maximally entangled state $|\Phi\rangle = \frac{1}{\sqrt{2^n}}\sum_j |j\rangle_A|j\rangle_B$, where $\langle\Phi|(I_A \otimes U_B)|\Phi\rangle = \operatorname{Tr}(U)/2^n$, but this doubles the qubit count.
\end{remark}

\subsection{Hilbert space structure}

\begin{remark}[Computational encoding]
\label{rem:computational-encoding}
Since $\mathcal{W}_t$ and the Hilbert space $\mathcal{H} = \mathcal{H}_p \otimes \mathcal{H}_s$ are isomorphic as $N^2$-dimensional complex vector spaces, with bases naturally indexed by $\mathbb{Z}_N^2$, we encode Weyl basis elements as computational basis states via $e_{p,s} \mapsto |p,s\rangle$. Under this encoding, left multiplication operators $L_{e_{p,s}}$ become phased permutations (Proposition~\ref{prop:weyl-mult}), efficiently implementable as quantum gates.
\end{remark}

We now characterize the subspace corresponding to the symmetric subalgebra.

\begin{proposition}[Symmetric subspace]
\label{prop:symmetric-subspace}
The symmetric skein elements correspond to quantum states in $\mathcal{H}_p \otimes \mathcal{H}_s$:
\[
|\psi_{(p,s)}\rangle := \frac{1}{\sqrt{2}}(|p,s\rangle + |-p,-s\rangle) \quad \text{for } (p,s) \neq (0,0),
\]
with $|\psi_{(0,0)}\rangle := |0,0\rangle$. These span a subspace
    $\mathcal{H}_{\text{Skein}} \subset \mathcal{H}_p \otimes \mathcal{H}_s$ of
    dimension $(N^2+1)/2$ for odd $N$.
\end{proposition}

\begin{proof}
The symmetric subalgebra $\mathcal{W}_t^{\iota}$ is spanned by elements $B_{(p,s)} = e_{p,s} + e_{-p,-s}$ with $B_{(p,s)} = B_{(-p,-s)}$ (Theorem~\ref{thm:fg-embedding}). The computational encoding (Remark~\ref{rem:computational-encoding}) maps $e_{p,s} \mapsto |p,s\rangle$, so the corresponding normalized quantum states are $|\psi_{(p,s)}\rangle = \frac{1}{\sqrt{2}}(|p,s\rangle + |-p,-s\rangle)$ for $(p,s) \neq (0,0)$, and $|\psi_{(0,0)}\rangle = |0,0\rangle$.
Notice that these states are orthogonal: $\langle \psi_{(p,s)} | \psi_{(p',s')} \rangle \neq 0$ if and only if $(p',s') = \pm(p,s)$, since the computational basis components of $|\psi_{(p,s)}\rangle$ are $\{|p,s\rangle, |-p,-s\rangle\}$. In particular, $|\psi_{(p,s)}\rangle = |\psi_{(-p,-s)}\rangle$, and $|\psi_{(0,0)}\rangle = |0,0\rangle$ is orthogonal to all states with $(p,s) \neq (0,0)$.

The involution $(p,s) \mapsto (-p,-s)$ on $\mathbb{Z}_N^2$ partitions the $N^2$ elements into orbits. For odd $N$, only $(0,0)$ is fixed (since $2p \equiv 0 \pmod{N}$ implies $p = 0$). The remaining $N^2 - 1$ elements form $(N^2-1)/2$ pairs. Taking one representative per orbit gives $(N^2+1)/2$ orthonormal basis vectors for $\mathcal{H}_{\text{Skein}}$, consistent with Corollary~\ref{cor:odd-level}.
\end{proof}

\subsection{Circuit implementation of the modular action}

The WRT invariant $Z_N(M_g; x_1, \ldots, x_m) = \operatorname{Tr}(\rho(g) L_{x_m} \cdots L_{x_1})$ requires implementing the modular action $\rho(g)$ for the monodromy $g \in SL_2(\mathbb{Z})$. We express $g$ as a word in the standard generators $S = \bigl(\begin{smallmatrix}0&-1\\1&0\end{smallmatrix}\bigr)$ and $T = \bigl(\begin{smallmatrix}1&1\\0&1\end{smallmatrix}\bigr)$, so it suffices to implement $\rho(S)$ and $\rho(T)$.

The Weil (metaplectic) representation provides explicit formulas \cite{scheithauer2009}. On $\mathcal{H}$ with computational basis $|p,s\rangle$ ($p,s \in \mathbb{Z}_N$):\footnote{For intuition: dropping phases yields index maps $(p,s)\mapsto(-s,p)$ for $S$ and $(p,s)\mapsto(p{+}s,s)$ for $T$. We retain all phases in circuit implementations.}
\begin{align}
\rho(S)\,|p,s\rangle\;&=\;\frac{1}{N}\sum_{r,u\in\mathbb{Z}_N} t^{\,2(pu-sr)}\,|r,u\rangle,\label{eq:WeilS-global}\\
\rho(T)\,|p,s\rangle\;&=\;t^{\,s^2}\,|\,p{+}s,\,s\rangle,\label{eq:WeilT-global}
\end{align}
where $t = e^{2\pi i/N}$. This is a projective representation of $SL_2(\mathbb{Z})$, realizing the modular functor structure of WRT theory \cite{reshetikhin1991}.

\begin{remark}[Projectivity]
The Weil representation is projective: $\rho(S)$ and $\rho(T)$ satisfy $\rho(S)^4 = I$ but only generate $SL_2(\mathbb{Z})$ up to scalar phases \cite{scheithauer2009}. This ambiguity does not affect our computation since the trace $\operatorname{Tr}(\rho(g) L_{x_m} \cdots L_{x_1})$ is invariant under $\rho(\cdot) \mapsto e^{i\phi}\rho(\cdot)$.
\end{remark}

We now translate these formulas into quantum circuits.

\begin{lemma}[Modular generator circuits]
\label{lem:modular-action}
For odd $N$ with $t = e^{2\pi i/N}$, there are depth-$O(\log^2 N)$\footnote{This bound is conservative; more precisely, the $T$-generator requires depth $O(\log N \cdot \log\log N)$ for modular multiplication using logarithmic-depth carry-lookahead adders \cite{draper2004}, while the $S$-generator requires depth $O(\log N)$ for the double QFT. We use $O(\log^2 N)$ as an upper bound throughout.}
circuits realizing these actions:
\[
\rho(S)\;=\;(\mathrm{QFT}_N\!\otimes\!\mathrm{QFT}_N)\,\circ\,(\mathrm{NEG}\!\otimes\! I)\,\circ\,\mathrm{SWAP},
\qquad
\rho(T)\;=\;\mathrm{QPHASE}_{s^2}\,\circ\,\mathrm{ADD}_{p\leftarrow p+s},
\]
where $\mathrm{QPHASE}_{s^2}:|p,s\rangle\mapsto t^{\,s^2}|p,s\rangle$ is the quadratic phase gate,
$\mathrm{ADD}_{p\leftarrow p+s}:|p,s\rangle\mapsto|p{+}s,s\rangle$,
$\mathrm{NEG}:|x\rangle\mapsto|-x\rangle$, and $\mathrm{QFT}_N$ is the QFT modulo $N$.
Both circuits have depth $O(\log^2 N)$.
\end{lemma}

\begin{proof}
Since $N$ is odd, $t^2$ is also a primitive $N$-th root of unity, so we define $\mathrm{QFT}_N|j\rangle := \frac{1}{\sqrt{N}}\sum_{k \in \mathbb{Z}_N} t^{2jk}|k\rangle$.

For the $S$-generator, $(\mathrm{NEG} \otimes I) \circ \mathrm{SWAP}$ maps $|p,s\rangle \mapsto |-s,p\rangle$, and applying $\mathrm{QFT}_N$ to each register yields
\[
(\mathrm{QFT}_N \otimes \mathrm{QFT}_N)|-s,p\rangle = \frac{1}{N} \sum_{r,u \in \mathbb{Z}_N} t^{-2sr} t^{2pu} |r,u\rangle = \frac{1}{N} \sum_{r,u \in \mathbb{Z}_N} t^{2(pu-sr)} |r,u\rangle,
\]
matching \eqref{eq:WeilS-global}. For the $T$-generator, $\mathrm{ADD}_{p \leftarrow p+s}$ maps $|p,s\rangle \mapsto |p+s,s\rangle$, and the quadratic phase then gives $t^{s^2}|p+s,s\rangle$, matching \eqref{eq:WeilT-global}. The depth bounds follow from standard modular arithmetic primitives \cite{draper2004,rines2018}.
\end{proof}

\subsection{Circuit implementation of skein multiplication}

The WRT formula also requires implementing skein multiplication operators $L_{B_{(p,s)}}$. By Theorem~\ref{thm:lcu-decomp}, each such operator decomposes as $L_{B_{(p,s)}} = U_{p,s} + U_{-p,-s}$ where $U_{p,s}$ are phased-permutation unitaries. We first give efficient circuits for $U_{p,s}$, then assemble them via linear combination of unitaries (LCU).

\begin{lemma}[Unitary factorization and circuit depth]
\label{lem:unitary-factorization}
Each phased-permutation unitary $U_{p,s}$ from Theorem~\ref{thm:lcu-decomp} factors as $U_{p,s} = \Phi_{(p,s)} \circ P_{(p,s)}$ where:
\begin{itemize}
\item $P_{(p,s)}$ is the cyclic permutation $(r,u) \mapsto (r+p \bmod N, u+s \bmod N)$, implementable in depth $O(\log N)$ using QFT-based modular addition \cite{draper2000}.
\item $\Phi_{(p,s)}$ is the diagonal phase operator $|r,u\rangle \mapsto t^{pu-sr} |r,u\rangle$, implementable in depth $O(\log^2 N)$ via modular multiplication \cite{rines2018}.
\end{itemize}
The overall circuit depth for $U_{p,s}$ is $O(\log^2 N)$.
\end{lemma}

\begin{proof}
The factorization follows directly from the definition of $U_{p,s}$:
\[
U_{p,s}|r,u\rangle = t^{pu-sr}|r+p, u+s\rangle = \Phi_{(p,s)} \circ P_{(p,s)} |r,u\rangle,
\]
where the phase uses the original coordinates $(r,u)$, so $\Phi_{(p,s)}$ must be applied after $P_{(p,s)}$. The permutation requires two parallel modular additions, each implementable in depth $O(\log N)$ \cite{draper2000}. The phase operator computes $pu - sr \bmod N$ and applies the corresponding rotation, requiring depth $O(\log^2 N)$ \cite{draper2004,rines2018}. Composing these yields total depth $O(\log^2 N)$.
\end{proof}

With circuits for $U_{p,s}$ established, we assemble them into an LCU implementation of the full skein multiplication operator.

\begin{theorem}[LCU implementation]
\label{thm:lcu-circuit}
The operator $L_{B_{(p,s)}} = U_{p,s} + U_{-p,-s}$ admits a one-ancilla two-branch LCU realizing a $\tfrac{1}{2}$ block-encoding:
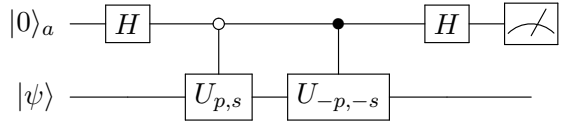
\[
\Qcircuit @C=1.2em @R=1.0em {
    \lstick{|0\rangle_a} & \gate{H} & \ctrlo{1} & \ctrl{1} & \gate{H} & \meter \\
    \lstick{|\psi\rangle} & \qw & \gate{U_{p,s}} & \gate{U_{-p,-s}} & \qw & \qw
}
\]
\end{theorem}

\begin{proof}
We verify the block-encoding by tracing the circuit's action on an arbitrary state:
\begin{align}
|0\rangle_a |\psi\rangle &\xrightarrow{H \otimes I} \frac{1}{\sqrt{2}}(|0\rangle_a + |1\rangle_a)|\psi\rangle \\
&\xrightarrow{\text{Controlled}} \frac{1}{\sqrt{2}}(|0\rangle_a U_{p,s}|\psi\rangle + |1\rangle_a U_{-p,-s}|\psi\rangle) \\
&\xrightarrow{H \otimes I} \frac{1}{2}|0\rangle_a(U_{p,s} + U_{-p,-s})|\psi\rangle + \frac{1}{2}|1\rangle_a(U_{p,s} - U_{-p,-s})|\psi\rangle.
\end{align}
The final Hadamard interferes the two branches, placing the desired sum $L_{B_{(p,s)}} = U_{p,s} + U_{-p,-s}$ in the $|0\rangle_a$ component with coefficient $1/2$. Post-selection thus yields $L_{B_{(p,s)}}|\psi\rangle/2$ with success probability $\|L_{B_{(p,s)}}|\psi\rangle\|^2/4$.
\end{proof}

\subsection{Total complexity}
\label{sec:total-complexity}

The quantum algorithm estimates the WRT invariant through repeated Hadamard test measurements. The total cost depends on two factors: the depth of each circuit execution, and the number of independent executions needed to achieve a target precision. We establish each in turn, then combine them.

\begin{theorem}[Circuit depth per execution]
\label{thm:complexity}
The quantum circuit implementing $W = \rho(g) L_{x_m} \cdots L_{x_1}$ has depth $D = O(\ell \log^2 N + m\log^2 N)$, where $\ell = |g|_{S,T}$ is the word length of $g$ in generators $S, T$.
\end{theorem}

\begin{proof}
The modular action $\rho(g)$ requires $\ell$ applications of $S$ or $T$ generators, each implementable in depth $O(\log^2 N)$ by Lemma~\ref{lem:modular-action}, yielding $O(\ell \log^2 N)$. The $m$ skein multiplications each require one LCU block with phased-permutation unitaries of depth $O(\log^2 N)$ by Lemma~\ref{lem:unitary-factorization}, yielding $O(m\log^2 N)$. The total is $D = O(\ell \log^2 N + m\log^2 N)$.
\end{proof}

Each circuit execution produces one measurement outcome. We now determine how many independent samples suffice to estimate the trace to precision $\epsilon$ with high probability.

\begin{proposition}[Number of samples]
\label{prop:sample-complexity}
Let $\widetilde{Z} = \operatorname{Tr}(W)/N^2$ denote the normalized trace. Estimating $\widetilde{Z}$ to additive precision $\epsilon$ with success probability at least $1 - \delta$ requires $n = O(\log(1/\delta)/\epsilon^2)$ independent circuit executions.
\end{proposition}

\begin{proof}
Each Hadamard test measurement (Proposition~\ref{prop:hadamard-test}) yields $X_i \in \{0,1\}$ with $\mathbb{E}[X_i] = \frac{1}{2}(1 + \text{Re}[\widetilde{Z}])$. By Hoeffding's inequality, the empirical mean $\bar{X}_n$ satisfies
\[
\Pr\bigl(|\bar{X}_n - \mathbb{E}[X_i]| \geq \epsilon/2\bigr) \leq 2\exp(-n\epsilon^2/2).
\]
Setting the right-hand side to $\delta$ and solving: $n = 2\log(2/\delta)/\epsilon^2 = O(\log(1/\delta)/\epsilon^2)$. The estimator $2\bar{X}_n - 1$ then approximates $\text{Re}[\widetilde{Z}]$ to within $\epsilon$. Estimating $\text{Im}[\widetilde{Z}]$ via a phase gate requires an equal number of samples, doubling the total without changing asymptotics.
\end{proof}

Multiplying circuit depth by sample count yields the total quantum cost. The following theorem summarizes both the sampling-based approach (suitable for near-term hardware) and the quantum amplitude estimation approach (optimal for fault-tolerant hardware).

\begin{theorem}[Total complexity for WRT computation]
\label{thm:total-complexity}
Let $M_g$ be a torus bundle with monodromy $g \in SL_2(\mathbb{Z})$ of word length $\ell$, and let $x_1, \ldots, x_m \in K_N(\Sigma_{1,0})$ be skein elements. The WRT invariant $Z_N(M_g; x_1, \ldots, x_m)$ can be approximated to additive precision $\epsilon$ with probability at least $1-\delta$ on $2\lceil\log_2 N\rceil + m + 1$ qubits. Let $D = O(\ell \log^2 N + m\log^2 N)$ denote the base circuit depth (Theorem~\ref{thm:complexity}). Two approaches achieve this:
\begin{enumerate}
\item \textbf{Sampling} (near-term): $O(\log(1/\delta)/\epsilon^2)$ independent runs, each of depth $D$, for total operations
\[
O\bigl(D \cdot \log(1/\delta)/\epsilon^2\bigr) = O\bigl((\ell \log^2 N + m\log^2 N) \cdot \log(1/\delta)/\epsilon^2\bigr).
\]
\item \textbf{Quantum amplitude estimation}~\cite{brassard2002} (fault-tolerant): $O(\log(1/\delta))$ runs, each of depth $O(D/\epsilon)$, for total operations
\[
O\bigl(D \cdot \log(1/\delta)/\epsilon\bigr) = O\bigl((\ell \log^2 N + m\log^2 N) \cdot \log(1/\delta)/\epsilon\bigr).
\]
\end{enumerate}
The sampling approach uses independent shallow circuits; QAE applies the base circuit $O(1/\epsilon)$ times coherently within each run via Grover-style reflections, requiring longer coherence times.
\end{theorem}

\begin{proof}
The base circuit implementing $W = \rho(g) L_{x_m} \cdots L_{x_1}$ has depth $D = O(\ell \log^2 N + m\log^2 N)$ (Theorem~\ref{thm:complexity}). The qubit count comprises $\lceil\log_2 N\rceil$ qubits for each of the $p$ and $s$ registers encoding $\mathbb{Z}_N^2$, $m$ ancillas for the LCU blocks (one per skein multiplication, as composing LCU circuits with a shared ancilla does not yield the correct block-encoding of the product---see Theorem~\ref{thm:lcu-circuit}), and one clean qubit for the Hadamard test (Proposition~\ref{prop:hadamard-test}).

For the sampling approach, each Hadamard test measurement yields an independent sample. By Hoeffding's inequality (Proposition~\ref{prop:sample-complexity}), $O(\log(1/\delta)/\epsilon^2)$ samples suffice to estimate the trace to precision $\epsilon$ with confidence $1-\delta$. Each sample requires one circuit execution of depth $D$, yielding total operations $O(D \cdot \log(1/\delta)/\epsilon^2)$.

For quantum amplitude estimation~\cite{brassard2002}, precision $\epsilon$ is achieved by applying the base circuit (and Grover reflections) $O(1/\epsilon)$ times coherently within each run, yielding circuit depth $O(D/\epsilon)$ per run. For confidence $1-\delta$, $O(\log(1/\delta))$ independent runs suffice. The total operation count is thus $O(D \cdot \log(1/\delta)/\epsilon)$---a quadratic improvement in $\epsilon$ at the cost of requiring $O(1/\epsilon)$ coherent controlled repetitions rather than independent executions.
\end{proof}

\begin{remark}[Hardware trade-offs]
\label{rem:amplitude-estimation}
The sampling approach suits near-term NISQ hardware: each circuit execution is independent, enabling straightforward error mitigation and requiring only short coherence times. Quantum amplitude estimation achieves quadratically better precision scaling but demands fault-tolerant hardware capable of maintaining coherence through $O(1/\epsilon)$ sequential applications of the base circuit with Grover-style reflections.
\end{remark}

\subsection{Complete quantum algorithm}

We now assemble the circuit components into a complete algorithm. The pseudocode below synthesizes the modular action (Lemma~\ref{lem:modular-action}), LCU decomposition (Theorem~\ref{thm:lcu-circuit}), and Hadamard test (Proposition~\ref{prop:hadamard-test}); the complexity analysis appears in Theorem~\ref{thm:total-complexity}.

\begin{algorithm}
\caption{Quantum Algorithm for WRT Invariants of Torus Bundles}
\label{alg:wrt-quantum}
\begin{algorithmic}[1]
\REQUIRE Monodromy $g \in \mathrm{SL}_2(\mathbb{Z})$, skein elements $x_1, \ldots, x_m \in K_N(\Sigma_{1,0})$, precision $\epsilon$, confidence $1-\delta$
\ENSURE Estimate of $Z_N(M_g; x_1, \ldots, x_m)/N^2$ within additive error $\epsilon$
\STATE Decompose $g = S^{a_1} T^{b_1} \cdots S^{a_k} T^{b_k}$ with word length $\ell$ (classical preprocessing)
\STATE Set sample count $n \gets O(\log(1/\delta)/\epsilon^2)$ (Proposition~\ref{prop:sample-complexity})
\FOR{$j = 1, \ldots, n$}
    \STATE Initialize $2\lceil\log_2 N \rceil + m + 1$ qubits (data registers, $m$ LCU ancillas, Hadamard ancilla)
    \FOR{$i = 1, \ldots, m$}
        \STATE Apply LCU circuit for $L_{x_i} = U_{p_i,s_i} + U_{-p_i,-s_i}$ on ancilla $a_i$ (Theorem~\ref{thm:lcu-circuit})
    \ENDFOR
    \STATE Apply modular action $\rho(g)$ via generator circuits (Lemma~\ref{lem:modular-action})
    \STATE Measure via Hadamard test (Proposition~\ref{prop:hadamard-test}), record $X_j \in \{0,1\}$
\ENDFOR
\RETURN $2\bar{X}_n - 1$ estimates $\mathrm{Re}[\widetilde{Z}]$; repeat with phase gate for $\mathrm{Im}[\widetilde{Z}]$
\end{algorithmic}
\textit{Each iteration has depth $O(\ell \log^2 N + m\log^2 N)$; total: $O\bigl((\ell \log^2 N + m\log^2 N) \cdot \log(1/\delta)/\epsilon^2\bigr)$ operations.}
\end{algorithm}

\begin{theorem}[Algorithm correctness]
\label{thm:algorithm-correctness}
Algorithm~\ref{alg:wrt-quantum} outputs an estimate of the normalized WRT invariant $Z_N(M_g; x_1, \ldots, x_m)/N^2$ with additive error at most $\epsilon$ and success probability at least $1-\delta$ on $2\lceil\log_2 N\rceil + m + 1$ qubits. The pseudocode shows the sampling-based approach; Theorem~\ref{thm:total-complexity} gives both the sampling complexity $O(D \cdot \log(1/\delta)/\epsilon^2)$ and the QAE complexity $O(D \cdot \log(1/\delta)/\epsilon)$ where $D = O(\ell \log^2 N + m\log^2 N)$ is the base circuit depth.
\end{theorem}

\begin{proof}
The algorithm constructs $W = \rho(g) L_{x_m} \cdots L_{x_1}$ where $L_{x_i} = L_{B_{(p_i,s_i)}}$ denotes left multiplication by the skein element $x_i$. Each multiplication operator decomposes as $L_{B_{(p_i,s_i)}} = U_{p_i,s_i} + U_{-p_i,-s_i}$ (Theorem~\ref{thm:lcu-decomp}), implemented exactly by the two-branch LCU circuit (Theorem~\ref{thm:lcu-circuit}). The modular action $\rho(g)$ factors through the generators $S$ and $T$, each realized by Lemma~\ref{lem:modular-action}. The Hadamard test (Proposition~\ref{prop:hadamard-test}) estimates $\operatorname{Tr}(W)/N^2$, which equals the normalized WRT invariant by definition. Correctness follows from the exact implementation of each component; the complexity analysis in Theorem~\ref{thm:total-complexity} derives both the sampling and QAE bounds.
\end{proof}

The key achievement is exponential space improvement: $O(\log N + m)$ qubits versus $\Theta(N^2)$ classical memory, while maintaining polynomial time. We now turn to a related but distinct computational problem where quantum superposition provides not merely space efficiency, but genuine computational advantage over classical methods. The shift from "computing a trace" to "exactly extracting a specific coefficient" dramatically changes the complexity landscape.

\section{FG--Coefficient Counting in the Colinear Case: Hardness and a Space-Efficient Quantum Routine}
\label{sec:fg-colinear}

While computing WRT invariants (traces) is polynomial-time, exactly extracting specific coefficients from skein products presents different complexity. We now formalize this distinction and prove that exact coefficient counting is \#P-complete.

\subsection{Motivation and novelty of the coefficient counting problem}

The coefficient extraction problem arises naturally from the Frohman--Gelca product rule and presents a fundamentally different computational challenge from the trace computation. While WRT invariants are polynomial-time computable, we show that exactly counting specific coefficients is \#P-complete (Theorem~\ref{thm:fgcc-hard}), implying no polynomial-time classical algorithm exists unless the polynomial hierarchy collapses \cite{toda1991}.

We complement this classical hardness with a specialized quantum routine (Theorem~\ref{thm:fgcc-quantum}) that adapts the LCU framework of Section~\ref{sec:quantum} to extract matrix elements via a Hadamard test. This yields a computational separation: while exact counting is \#P-complete, coefficients can be estimated to additive precision in polynomial time using $O(\log N)$ qubits. This parallels the Jones polynomial, where exact computation and multiplicative approximation are \#P-hard \cite{jaeger1990,kuperberg2015}, yet quantum algorithms achieve efficient additive approximation \cite{aharonov2009}. Unlike some representation-theoretic coefficient problems where conjectured quantum advantages were later matched by classical algorithms \cite{classical_quantum_multiplicities_2025}, the \#P-completeness of our problem provides a provable hardness guarantee for exact computation. The connection to \#SIGNED-SUM, a variant of the \#P-complete \#SUBSET-SUM problem \cite{valiant1979}, places this problem within classical counting complexity.

The Frohman--Gelca (FG) product-to-sum rule on the torus says that
\[
B_{(p,s)}\,B_{(r,u)}
\;=\;
t^{\,pu-sr}\,B_{(p+r,s+u)} \;+\; t^{-(pu-sr)}\,B_{(p-r,s-u)},
\qquad t=e^{2\pi i/N}.
\]
When all inputs are \emph{colinear}, i.e.\ of the form $v_i=(a_i,0)$, the symplectic form $\langle v_i,v_j\rangle = a_i\cdot 0 - 0\cdot a_j$ vanishes and thus all FG phases are $1$. Iterating the rule then yields a pure, unweighted binary expansion over sign assignments.

\begin{remark}[Encoding convention]
\label{rem:encoding}
Throughout this section, all integers are given in standard binary; the input size is the total bit-length $\sum_i \lceil \log_2(|a_i|{+}1)\rceil + \lceil \log_2(|z|{+}1)\rceil$.
\end{remark}

\begin{definition}[FG--Coefficient Counting (colinear)]
\label{def:fgcc}
Given integers $a_1,\ldots,a_m$ and a target $z$, expand the product $\prod_{i=1}^{m} B_{(a_i,0)}$ under the FG rule with indices treated as integers (not reduced modulo $N$). Define $c(z)$ to be the coefficient of $B_{(z,0)}$ in this expansion. The \emph{FG--Coefficient Counting} problem asks to compute $c(z)$ exactly.
\end{definition}

In the colinear regime one has the explicit identity
\begin{equation}\label{eq:fg-colinear-identity}
\prod_{i=1}^{m} B_{(a_i,0)}
\;=\;
\sum_{\varepsilon\in\{\pm 1\}^{m}}
B_{\left(\sum_{i=1}^m \varepsilon_i a_i,\;0\right)}
\quad\Longrightarrow\quad
c(z)\;=\;\#\Big\{\varepsilon\in\{\pm 1\}^{m}:\ \sum_{i=1}^m \varepsilon_i a_i = z\Big\}.
\end{equation}

This identity reveals that the coefficient $c(z)$ counts the number of ways to achieve the target sum $z$ using signed combinations of the inputs $a_1,\ldots,a_m$. This is precisely the \#SIGNED-SUM counting problem (a signed variant of \#SUBSET-SUM where every element must be included with a chosen sign), which suggests that coefficient extraction may be computationally hard. We now make this connection rigorous.

\begin{figure}[t!]
\centering
\begin{tikzpicture}[
    font=\small
]

\begin{scope}[shift={(0,0)}]
    \node[font=\bfseries, text=gruvbox-red] at (0,3.5) {Classical};

    \coordinate (root) at (0,2.8);
    \fill[gruvbox-gray] (root) circle (2.5pt);
    \node[font=\tiny, above] at (root) {0};

    \coordinate (l1a) at (-0.9,2.1);
    \coordinate (l1b) at (0.9,2.1);
    \fill[gruvbox-gray] (l1a) circle (2.5pt);
    \fill[gruvbox-gray] (l1b) circle (2.5pt);
    \draw[gruvbox-gray, thick] (root) -- (l1a) node[midway, left, font=\tiny, gruvbox-red] {$+a_1$};
    \draw[gruvbox-gray, thick] (root) -- (l1b) node[midway, right, font=\tiny, gruvbox-red] {$-a_1$};
    \node[font=\tiny, right] at (2.0,2.1) {$2^1{=}2$ paths};

    \coordinate (l2a) at (-1.3,1.4);
    \coordinate (l2b) at (-0.45,1.4);
    \coordinate (l2c) at (0.45,1.4);
    \coordinate (l2d) at (1.3,1.4);
    \fill[gruvbox-gray] (l2a) circle (2.5pt);
    \fill[gruvbox-gray] (l2b) circle (2.5pt);
    \fill[gruvbox-gray] (l2c) circle (2.5pt);
    \fill[gruvbox-gray] (l2d) circle (2.5pt);
    \draw[gruvbox-gray, thick] (l1a) -- (l2a) node[midway, left, font=\tiny, gruvbox-red] {$+a_2$};
    \draw[gruvbox-gray, thick] (l1a) -- (l2b) node[midway, right, font=\tiny, gruvbox-red] {$-a_2$};
    \draw[gruvbox-gray, thick] (l1b) -- (l2c) node[midway, left, font=\tiny, gruvbox-red] {$+a_2$};
    \draw[gruvbox-gray, thick] (l1b) -- (l2d) node[midway, right, font=\tiny, gruvbox-red] {$-a_2$};
    \node[font=\tiny, right] at (2.0,1.4) {$2^2{=}4$ paths};

    \coordinate (l3a) at (-1.4,0.7);
    \coordinate (l3b) at (-1.0,0.7);
    \coordinate (l3c) at (-0.6,0.7);
    \coordinate (l3d) at (-0.2,0.7);
    \coordinate (l3e) at (0.2,0.7);
    \coordinate (l3f) at (0.6,0.7);
    \coordinate (l3g) at (1.0,0.7);
    \coordinate (l3h) at (1.4,0.7);
    \foreach \x in {-1.4,-1.0,-0.6,-0.2,0.2,0.6,1.0,1.4} {
        \fill[gruvbox-gray] (\x,0.7) circle (2pt);
    }
    \draw[gruvbox-gray, thick] (l2a) -- (l3a) node[midway, left, font=\tiny, gruvbox-red] {$+a_3$};
    \draw[gruvbox-gray, thick] (l2a) -- (l3b) node[midway, right, font=\tiny, gruvbox-red] {$-a_3$};
    \draw[gruvbox-gray, thick] (l2b) -- (l3c);
    \draw[gruvbox-gray, thick] (l2b) -- (l3d);
    \draw[gruvbox-gray, thick] (l2c) -- (l3e);
    \draw[gruvbox-gray, thick] (l2c) -- (l3f);
    \draw[gruvbox-gray, thick] (l2d) -- (l3g);
    \draw[gruvbox-gray, thick] (l2d) -- (l3h);
    \node[font=\tiny, right] at (2.0,0.7) {$2^3{=}8$ paths};

    \node[font=\large] at (0,0.2) {$\vdots$};

    \node[font=\footnotesize] at (0,-0.2) {Level $m$};
    \node[font=\tiny, right] at (2.0,-0.2) {$2^m$ paths};

    \foreach \x in {-1.0, 0.2, 1.0} {
        \fill[gruvbox-green, thick] (\x, -0.7) circle (3.5pt);
    }
    \node[font=\small, gruvbox-green, align=center] at (0,-1.1) {\textbf{Output:} $c(z)$};
\end{scope}

\begin{scope}[shift={(6.5,0)}]
    \node[rectangle, draw=gruvbox-red, very thick, fill=gruvbox-light-red!10, text width=4.2cm, align=left, rounded corners] at (0,1.5) {
        \textbf{Classical Enumeration} \\[3pt]
        $\bullet$ Enumerate all $2^m$ sign \\
        \phantom{$\bullet$} assignments $\varepsilon \in \{\pm 1\}^m$ \\
        $\bullet$ For each: compute \\
        \phantom{$\bullet$} $\sum_i \varepsilon_i a_i$ \\
        $\bullet$ Count how many equal $z$ \\[5pt]
        \textbf{Complexity:} $\Omega(2^m)$ \\
        \textcolor{gruvbox-red}{\textbf{\#P-complete}}
    };
\end{scope}

\begin{scope}[shift={(0,-6)}]
    \node[font=\bfseries, text=gruvbox-blue] at (0,3.5) {Quantum};

    \node[font=\small, align=center] at (0,2.7) {
        Hadamard test
    };

    \draw[very thick] (-1.5,1.8) -- (1.8,1.8);
    \draw[very thick] (-1.5,1.0) -- (1.8,1.0);
    \node[left, font=\small] at (-1.5,1.8) {$|0\rangle$};
    \node[left, font=\small] at (-1.5,1.0) {$|0\rangle$};

    \draw[fill=gruvbox-light-blue!30, thick] (-0.9,1.6) rectangle (-0.5,2.0);
    \node[font=\small] at (-0.7,1.8) {H};

    \fill (0.1,1.8) circle (2.5pt);
    \draw[very thick] (0.1,1.8) -- (0.1,1.0);
    \draw[fill=gruvbox-light-blue!30, thick] (-0.3,0.75) rectangle (0.5,1.25);
    \node[font=\small] at (0.1,1.0) {LCU};

    \draw[fill=gruvbox-light-blue!30, thick] (1.0,1.6) rectangle (1.4,2.0);
    \node[font=\small] at (1.2,1.8) {H};

    \node[right, font=\small] at (1.8,1.8) {measure};

    \node[font=\small, align=center] at (0,0.2) {
        LCU creates $2^m$-term \\ superposition
    };

    \node[font=\small, gruvbox-blue, align=center] at (0,-0.5) {
        \textbf{Output:} $\approx c(z)/2^m$
    };
\end{scope}

\begin{scope}[shift={(6.5,-6)}]
    \node[rectangle, draw=gruvbox-blue, very thick, fill=gruvbox-light-blue!10, text width=4.2cm, align=left, rounded corners] at (0,1.5) {
        \textbf{Quantum Hadamard Test} \\[3pt]
        $\bullet$ Apply Hadamard test with \\
        \phantom{$\bullet$} $U = \prod_i \mathcal{U}(a_i)$ \\
        $\bullet$ LCU product creates $2^m$ \\
        \phantom{$\bullet$} superposition terms \\
        $\bullet$ Sample $O(1/\epsilon^2)$ times \\
        $\bullet$ Estimate $c(z)/2^m$ \\[5pt]
        \textbf{Complexity:} $O(m \log^2 N / \epsilon^{2})$ \\
        \textcolor{gruvbox-blue}{\textbf{Polynomial-time}}
    };
\end{scope}

\end{tikzpicture}
\caption{Classical vs Quantum approaches to FG-Coefficient Counting. \textbf{Top (Classical):} Binary tree shows exponential explosion. Each node represents a partial sum, and each edge represents adding $+a_i$ or $-a_i$ (edge labels). Starting from sum 0, each choice doubles the number of paths: level 1 has $2^1{=}2$ paths, level 2 has $2^2{=}4$ paths, reaching $2^m$ total paths at level $m$. Each path from root to leaf corresponds to one sign assignment $\varepsilon \in \{\pm 1\}^m$. The classical algorithm must enumerate all $2^m$ paths to count which ones yield sum $z$ (green nodes)—this is \#SIGNED-SUM counting (a signed variant of \#SUBSET-SUM), requiring exponential $\Omega(2^m)$ time. \textbf{Bottom (Quantum):} Hadamard test evaluates all $2^m$ assignments simultaneously via superposition. Sampling estimates the normalized coefficient $\alpha = c(z)/2^m$ to additive precision $\epsilon$ in polynomial time $O(m \log^2 N / \epsilon^{2})$ (Theorem~\ref{thm:fgcc-quantum}) for common targets.}
\label{fig:coefficient_counting}
\end{figure}

\subsection{Complexity: \#P-completeness via a parsimonious reduction}
\label{subsec:fgcc-hardness}

To establish the computational hardness of coefficient counting, we reduce from a signed variant of the classical \#SUBSET-SUM problem. The standard \#SUBSET-SUM problem counts submultisets (choosing which elements to include), while the variant relevant to our skein products counts sign assignments (where every element must be included with either a + or $-$ sign). Define:
\[
\#\mathrm{SIGNED\text{-}SUM}(a_1,\ldots,a_m;z)
\;=\;
\#\Big\{\sigma\in\{\pm 1\}^{m}:\ \sum_i \sigma_i a_i = z\Big\}.
\]
This signed formulation is \#P-complete. To see this, observe that sign assignments $\sigma \in \{\pm 1\}^m$ correspond bijectively to subsets $S \subseteq [m]$ via $S = \{i : \sigma_i = +1\}$. Under this correspondence, $\sum_i \sigma_i a_i = 2\sum_{i \in S} a_i - \sum_i a_i$, so $\#\mathrm{SIGNED\text{-}SUM}(a_1,\ldots,a_m;z) = \#\mathrm{SUBSET\text{-}SUM}(a_1,\ldots,a_m; t)$ where $t = (z + \sum_i a_i)/2$ when this quantity is an integer, and zero otherwise. Since \#SUBSET-SUM is \#P-complete~\cite{valiant1979}, so is \#SIGNED-SUM.

\begin{theorem}[FG--Coefficient Counting is \#P-complete]
\label{thm:fgcc-hard}
The problem in Definition~\ref{def:fgcc} is \#P-complete under parsimonious reductions.
\end{theorem}

\begin{proof}
For membership in \#P, we exhibit a nondeterministic polynomial-time Turing machine $M$ whose accepting path count equals $c(z)$. On input $(a_1,\ldots,a_m,z)$, the machine $M$ nondeterministically guesses a sign assignment $\varepsilon \in \{\pm 1\}^m$ using $m$ nondeterministic bits, deterministically computes $\sum_{i=1}^m \varepsilon_i a_i$ in polynomial time, and accepts if and only if this sum equals $z$. The number of accepting paths is $\#\{\varepsilon \in \{\pm 1\}^m : \sum_i \varepsilon_i a_i = z\} = c(z)$ by equation~\eqref{eq:fg-colinear-identity}, establishing membership in \#P~\cite{valiant1979}.

For \#P-hardness, we give a parsimonious reduction from \#SIGNED-SUM, which is \#P-complete (as shown above). The reduction is simply the identity map $f:(a_1,\ldots,a_m,z) \mapsto (a_1,\ldots,a_m,z)$, which is trivially polynomial-time computable. Parsimony follows directly from equation~\eqref{eq:fg-colinear-identity}:
\[
\#\mathrm{SIGNED\text{-}SUM}(a_1,\ldots,a_m;z) = \#\{\sigma \in \{\pm 1\}^m : \textstyle\sum_i \sigma_i a_i = z\} = c(z).
\]
Thus FG--Coefficient Counting is \#P-hard under parsimonious reductions, completing the proof.
\end{proof}

\begin{corollary}[Decision version]
\label{cor:fgcc-decision}
The decision problem ``Is $c(z)>0$?'' for colinear inputs is NP-complete.
\end{corollary}

\begin{proof}
For membership in NP, a sign assignment $\varepsilon \in \{\pm 1\}^m$ serves as a polynomial-size certificate: given $\varepsilon$, one verifies $\sum_i \varepsilon_i a_i = z$ in polynomial time.

For NP-hardness, we reduce from SUBSET-SUM~\cite{garey1979}. Given integers $b_1, \ldots, b_m$ and target $T$, set $a_i = b_i$ and $z = 2T - \sum_i b_i$. A subset $S$ with $\sum_{i \in S} b_i = T$ corresponds to the sign assignment $\varepsilon_i = +1$ if $i \in S$ and $\varepsilon_i = -1$ otherwise, yielding $\sum_i \varepsilon_i a_i = 2\sum_{i \in S} b_i - \sum_i b_i = 2T - \sum_i b_i = z$. Conversely, any $\varepsilon$ with $\sum_i \varepsilon_i a_i = z$ determines $S = \{i : \varepsilon_i = +1\}$ satisfying $\sum_{i \in S} b_i = T$. Thus $c(z) > 0$ if and only if the SUBSET-SUM instance has a solution.
\end{proof}

\begin{lemma}[Integer-to-modular equivalence]
\label{lem:int-mod-equiv}
Let $c(z)$ denote the integer coefficient from Definition~\ref{def:fgcc}, and let $c_N(z)$ denote the number of sign assignments $\varepsilon \in \{\pm 1\}^m$ with $\sum_i \varepsilon_i a_i \equiv z \pmod N$. Set
\begin{equation}\label{eq:no-wrap-bound}
A \;:=\; \sum_{i=1}^m |a_i|, \qquad B \;:=\; A + |z|.
\end{equation}
If $N > B$, then $c_N(z) = c(z)$. Conversely, for any $N \geq 1$, there exist inputs where $N = B$ and $c_N(z) \neq c(z)$.
\end{lemma}

\begin{proof}
For the forward direction, every achievable sum $s = \sum_i \varepsilon_i a_i$ satisfies $|s| \leq A$. Suppose $s \equiv z \pmod N$. If $s \neq z$, then $|s - z| \geq N$ (as a nonzero multiple of $N$), yet $|s - z| \leq |s| + |z| \leq B$ by the triangle inequality. The hypothesis $N > B$ yields the contradiction $N \leq B < N$, so $s = z$. Thus $c_N(z) = c(z)$.

For the converse, take $m = 1$, $a_1 = N$, and $z = 0$. Then $A = B = N$, and the achievable sums $\{\pm N\}$ both satisfy $\pm N \equiv 0 \pmod N$, giving $c_N(0) = 2 \neq 0 = c(0)$.
\end{proof}

\begin{remark}[Fixed-$N$ DP vs.\ exact integer hardness]
\label{rem:wrap-vs-hardness}
Our torus-bundle DP (Section~\ref{alg:classical-dp}) is polynomial-time because it combines all $2^m$ branches \emph{modulo $N$} into an $N^2$ table. Theorem~\ref{thm:fgcc-hard} concerns the \emph{integer} coefficient $c(z)$ from Definition~\ref{def:fgcc}. By Lemma~\ref{lem:int-mod-equiv}, choosing $N > B$ (with $B$ as in~\eqref{eq:no-wrap-bound}) makes the two problems equivalent---but then the DP table has size pseudo-polynomial in the input magnitudes (exponential in the input bit-length), matching \#P-hardness.
\end{remark}

Having established that exact coefficient extraction is \#P-complete classically, we now demonstrate that quantum computation can efficiently \emph{approximate} this hard counting problem. While classical exact methods must enumerate exponentially many terms, the quantum approach encodes all $2^m$ sign assignments in superposition, achieving polynomial scaling in both circuit depth and qubit count.

\subsection{A quantum routine tailored to FG--Coefficient Counting}
\label{subsec:fgcc-quantum}

We now show how the two-term structure from our trace-estimation algorithm yields a compact quantum procedure to estimate $c(z)$. The circuit operates on a finite register $\mathbb{Z}_N$; by Lemma~\ref{lem:int-mod-equiv}, choosing $N > B$ (with $B$ as in~\eqref{eq:no-wrap-bound}) ensures the modular computation yields the correct integer coefficient $c(z)$. This requires only $O(\log N) = O(\log B)$ data qubits.

\begin{lemma}[Colinear LCU specialization]
\label{lem:colinear-lcu}
For colinear inputs $(p,s) = (a,0)$, the phased-permutation unitaries from Theorem~\ref{thm:lcu-decomp} reduce to pure modular translations on the $s=0$ subspace:
\[
U_{a,0}|r,0\rangle = |r+a, 0\rangle, \qquad U_{-a,0}|r,0\rangle = |r-a, 0\rangle.
\]
Let $\mathcal{U}_{a_i}$ denote the LCU circuit from Theorem~\ref{thm:lcu-circuit} specialized to $(a_i,0)$, each acting on the data register together with its own ancilla qubit $|0\rangle_{a_i}$. The composition using $m$ independent ancillas satisfies
\begin{equation}\label{eq:block-product}
\bra{0}^{\otimes m}_{\mathbf{a}} \Big(\prod_{i=1}^{m} \mathcal{U}_{a_i}\Big) \ket{0}^{\otimes m}_{\mathbf{a}} \;=\; 2^{-m} \prod_{i=1}^{m} L_{B_{(a_i,0)}},
\end{equation}
where $\bra{0}^{\otimes m}_{\mathbf{a}}$ denotes projection of all $m$ ancillas onto $|0\rangle$.
\end{lemma}

\begin{proof}
By Theorem~\ref{thm:lcu-decomp}, $U_{p,s}|r,u\rangle = t^{pu-sr}|r+p, u+s\rangle$. For $(p,s) = (a,0)$ acting on a state with $u=0$, the phase becomes $t^{a \cdot 0 - 0 \cdot r} = 1$, reducing $U_{a,0}$ and $U_{-a,0}$ to the stated translations.

By Theorem~\ref{thm:lcu-circuit}, each circuit $\mathcal{U}_{a_i}$ realizes a $\tfrac{1}{2}$ block-encoding: $\bra{0}_{a_i} \mathcal{U}_{a_i} \ket{0}_{a_i} = \tfrac{1}{2} L_{B_{(a_i,0)}}$. Since each circuit acts on its own ancilla, the projections factor over the tensor product structure:
\[
\bra{0}^{\otimes m}_{\mathbf{a}} \Big(\prod_{i=1}^{m} \mathcal{U}_{a_i}\Big) \ket{0}^{\otimes m}_{\mathbf{a}} = \prod_{i=1}^{m} \bra{0}_{a_i} \mathcal{U}_{a_i} \ket{0}_{a_i} = 2^{-m} \prod_{i=1}^{m} L_{B_{(a_i,0)}}.
\]
\end{proof}

\begin{proposition}[Coefficient extraction via Hadamard test]
\label{prop:coeff-extraction}
Let $\mathcal{U} := \prod_{i=1}^{m} \mathcal{U}_{a_i}$ denote the composition of $m$ LCU circuits, each with its own ancilla (Lemma~\ref{lem:colinear-lcu}). The coefficient $c(z)$ can be extracted using the Hadamard test for matrix elements~\cite{cleve1998}: prepare
\[
|\Omega\rangle \;:=\; \tfrac{1}{\sqrt{2}}\big(|0\rangle_c|0\rangle^{\otimes m}_{\mathbf{a}}|0,0\rangle + |1\rangle_c|0\rangle^{\otimes m}_{\mathbf{a}}|z,0\rangle\big),
\]
apply controlled-$\mathcal{U}$ (conditioned on $c=0$), then measure Pauli $X$ on the control qubit. The expectation satisfies
\begin{equation}\label{eq:coeff-extraction}
\mathbb{E}[X_c] \;=\; 2^{-m}\, c(z).
\end{equation}
\end{proposition}

The circuit implementing Proposition~\ref{prop:coeff-extraction} combines the Hadamard test for matrix elements with the LCU composition:
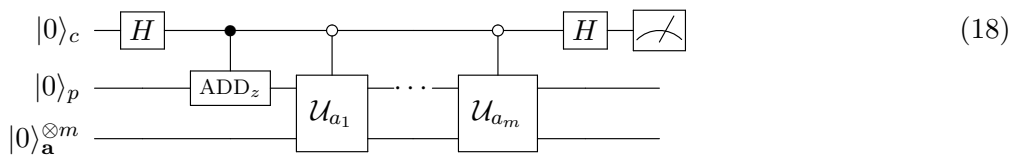
\begin{equation}\label{eq:coeff-circuit}
\Qcircuit @C=0.9em @R=0.7em {
    \lstick{|0\rangle_c} & \gate{H} & \ctrl{1} & \ctrlo{1} & \qw & \ctrlo{1} & \gate{H} & \meter \\
    \lstick{|0\rangle_p} & \qw & \gate{\scriptstyle\text{ADD}_z} & \multigate{1}{\mathcal{U}_{a_1}} & \push{\cdots\,} \qw & \multigate{1}{\mathcal{U}_{a_m}} & \qw & \qw \\
    \lstick{|0\rangle^{\otimes m}_{\mathbf{a}}} & \qw & \qw & \ghost{\mathcal{U}_{a_1}} & \qw & \ghost{\mathcal{U}_{a_m}} & \qw & \qw
}
\end{equation}
The Hadamard test for matrix elements extracts $\text{Re}\langle\psi|U|\phi\rangle$ by preparing $(|0\rangle|\phi\rangle + |1\rangle|\psi\rangle)/\sqrt{2}$, applying $U$ anti-controlled on $c$, then measuring $X$ on $c$. Here $|\phi\rangle = |0\rangle^{\otimes m}_{\mathbf{a}}|0\rangle_p$ and $|\psi\rangle = |0\rangle^{\otimes m}_{\mathbf{a}}|z\rangle_p$, so the $\text{ADD}_z$ gate (a controlled modular adder: $|p\rangle \mapsto |p+z\rangle$ when $c=1$, implementable via QFT-based arithmetic~\cite{draper2000}) prepares the required superposition after the initial Hadamard. The anti-controlled $\mathcal{U}_{a_i}$ blocks (each an LCU circuit from Theorem~\ref{thm:lcu-circuit} with its own ancilla) apply only when $c=0$. Post-selecting all $m$ ancillas onto $|0\rangle$ yields the block-encoding $2^{-m}\prod_i L_{B_{(a_i,0)}}$ by Lemma~\ref{lem:colinear-lcu}, and the Hadamard test extracts $\langle z|2^{-m}\prod_i L_{B_{(a_i,0)}}|0\rangle = 2^{-m}c(z)$.

\begin{proof}
The Hadamard test for matrix elements~\cite{cleve1998} yields $\mathbb{E}[X_c] = \text{Re}(\bra{0}^{\otimes m}_{\mathbf{a}}\bra{z,0}\,\mathcal{U}\,\ket{0}^{\otimes m}_{\mathbf{a}}\ket{0,0})$. By Lemma~\ref{lem:colinear-lcu}, $\bra{0}^{\otimes m}_{\mathbf{a}} \mathcal{U} \ket{0}^{\otimes m}_{\mathbf{a}} = 2^{-m} \prod_{i=1}^{m} L_{B_{(a_i,0)}}$. The two-term structure $L_{B_{(a,0)}} = U_{a,0} + U_{-a,0}$ (Theorem~\ref{thm:lcu-decomp}) implies that the product operator expands the binary tree of sign assignments:
\[
\Big(\prod_{i=1}^{m} L_{B_{(a_i,0)}}\Big) |0,0\rangle \;=\; \sum_{\varepsilon \in \{\pm 1\}^m} \Big|\sum_{i=1}^m \varepsilon_i a_i,\, 0\Big\rangle,
\]
where each $L_{B_{(a_i,0)}}$ branches the superposition into $|\cdot + a_i, 0\rangle$ and $|\cdot - a_i, 0\rangle$ components. Taking the inner product with $\langle z,0|$ counts sign assignments satisfying $\sum_i \varepsilon_i a_i = z$, which equals $c(z)$ by equation~\eqref{eq:fg-colinear-identity}. Thus $\bra{0}^{\otimes m}_{\mathbf{a}}\bra{z,0}\,\mathcal{U}\,\ket{0}^{\otimes m}_{\mathbf{a}}\ket{0,0} = 2^{-m} c(z)$. Since $c(z) \in \mathbb{Z}_{\geq 0}$, the imaginary part vanishes.
\end{proof}

The Hadamard test for matrix elements yields the normalized coefficient $\alpha = 2^{-m}c(z)$ in expectation (Proposition~\ref{prop:coeff-extraction}). The following theorem summarizes both approaches to achieving precision $\epsilon$ with confidence $1-\delta$, paralleling Theorem~\ref{thm:total-complexity}.

\begin{theorem}[Quantum Coefficient Estimation]
\label{thm:fgcc-quantum}
For $N > B$ (with $B$ as in~\eqref{eq:no-wrap-bound}), the normalized coefficient $\alpha := 2^{-m}c(z) \in [0,1]$ can be estimated to additive precision $\epsilon$ with probability at least $1-\delta$ on $2\lceil\log_2 N\rceil + m + 1$ qubits ($O(\log N)$ data qubits plus $m+1$ ancillas). Let $D = O(m\log^2 N)$ denote the base circuit depth. Two approaches achieve this:
\begin{enumerate}
\item \textbf{Sampling} (near-term): $O(\log(1/\delta)/\epsilon^2)$ independent runs, each of depth $D$, for total operations
\[
O\bigl(D \cdot \log(1/\delta)/\epsilon^2\bigr) = O\bigl(m\log^2 N \cdot \log(1/\delta)/\epsilon^2\bigr).
\]
\item \textbf{Quantum amplitude estimation}~\cite{brassard2002} (fault-tolerant): $O(\log(1/\delta))$ runs, each of depth $O(D/\epsilon)$, for total operations
\[
O\bigl(D \cdot \log(1/\delta)/\epsilon\bigr) = O\bigl(m\log^2 N \cdot \log(1/\delta)/\epsilon\bigr).
\]
\end{enumerate}
The sampling approach uses independent shallow circuits; QAE applies the base circuit $O(1/\epsilon)$ times coherently within each run via Grover-style reflections, requiring longer coherence times (see Remark~\ref{rem:amplitude-estimation}). Both approaches provide a space-efficient alternative to classical dynamic programming, using $O(\log N + m)$ qubits versus $O(N)$ classical memory.
\end{theorem}

\begin{proof}
The quantum circuit requires $2\lceil\log_2 N\rceil$ data qubits ($\lceil\log_2 N\rceil$ each for the $p$- and $s$-registers), $m$ LCU ancillas (one per colinear block, as required by Lemma~\ref{lem:colinear-lcu} to ensure correct block-encoding composition), and one Hadamard test control qubit, totaling $2\lceil\log_2 N\rceil + m + 1$ qubits. In the colinear case where $s=0$ throughout, the $s$-register can be omitted, reducing data qubits to $\lceil\log_2 N\rceil$.

For base circuit depth, each LCU block implementing $B_{(a_i,0)}$ consists of Hadamard gates on the ancilla (depth 1 each) and controlled-$U_{\pm a_i,0}$ operations. In the colinear case, the symplectic phase $t^{pu-sr}$ vanishes since $s=0$, so $U_{(\pm a,0)}:|p,0\rangle \mapsto |p \pm a, 0\rangle$ reduces to pure modular addition without phase computation. Each modular addition requires depth $O(\log N)$ via QFT-based adders~\cite{draper2000}. For consistency with the general (non-colinear) case where phases require $O(\log^2 N)$ depth~\cite{rines2018}, we state the conservative bound $O(\log^2 N)$ per LCU block, yielding base depth $D = O(m\log^2 N)$ for $m$ sequential blocks.

For the sampling approach, each Hadamard test measurement yields an independent sample. By Hoeffding's inequality (Proposition~\ref{prop:sample-complexity}), $O(\log(1/\delta)/\epsilon^2)$ samples suffice to estimate $\alpha$ to precision $\epsilon$ with confidence $1-\delta$. Each sample requires one circuit execution of depth $D$, yielding total operations $O(D \cdot \log(1/\delta)/\epsilon^2)$.

For quantum amplitude estimation~\cite{brassard2002}, precision $\epsilon$ is achieved by applying the base circuit (and Grover reflections) $O(1/\epsilon)$ times coherently within each run, yielding circuit depth $O(D/\epsilon)$ per run. For confidence $1-\delta$, $O(\log(1/\delta))$ independent runs suffice. The total operation count is thus $O(D \cdot \log(1/\delta)/\epsilon) = O(m\log^2 N \cdot \log(1/\delta)/\epsilon)$---a quadratic improvement in $\epsilon$.
\end{proof}

\begin{remark}[Quantum advantages and complexity barriers]
\label{rem:quantum-advantages}
The quantum algorithm provides exponential advantages in space and parallel time, with a more nuanced picture for serial time. For space, classical dynamic programming for \#SIGNED-SUM requires an array of size $O(N)$ to track counts for all achievable sums; with $m$ values encoded as $b$-bit integers, we have $B \leq (m+1) \cdot 2^b$, so $N = O(m \cdot 2^b)$ suffices to prevent wraparound—yielding $O(m \cdot 2^b)$ classical memory, exponential in $b$. The quantum algorithm uses $O(\log N + m) = O(b + m)$ qubits, polynomial in the input bit-length. For parallel time, classical enumeration requires $O(2^m)$ sequential steps while the quantum circuit has depth $O(m\log^2 N)$, polynomial in $m$ and $b$. Both advantages hold unconditionally.

For serial time, the picture depends on whether the target sum $z$ is common or rare. Recall from~\eqref{eq:fg-colinear-identity} that $c(z)$ counts how many of the $2^m$ sign configurations $\varepsilon \in \{\pm 1\}^m$ satisfy $\sum_i \varepsilon_i a_i = z$. A target is \emph{common} if a constant fraction of configurations achieve it ($c(z) = \Omega(2^m)$) and \emph{rare} if only a handful do ($c(z) = O(1)$).

For exact recovery, precision $\epsilon < 2^{-(m+1)}$ is required so that rounding the estimate of $\alpha = c(z)/2^m$ recovers the integer $c(z)$. This yields $O(2^m)$ quantum operations---matching classical enumeration. The exponential cost is unavoidable: exactly computing $c(z)$ solves a \#P-complete problem (Theorem~\ref{thm:fgcc-hard}), and polynomial-time algorithms would collapse the polynomial hierarchy~\cite{toda1991}.

For additive approximation, the quantum algorithm estimates $\alpha$ to precision $\epsilon$, incurring absolute error $\epsilon \cdot 2^m$ on $c(z)$. For common targets, this relative error is $O(\epsilon)$---meaningful approximation in polynomial time while classical exact counting remains $\Omega(2^m)$. For rare targets, the error overwhelms the signal and the approximation is uninformative.
\end{remark}

\section{Conclusions}

This work presents two complementary algorithms for computing Witten-Reshetikhin-Turaev invariants of torus bundles, exploiting the embedding of the skein algebra $K_N(\Sigma_{1,0})$ into the symmetric subalgebra of the non-commutative torus. First, we provide an explicit algorithmic formulation of the polynomial-time classical dynamic programming method with complete complexity analysis, interpreting the computation through the lens of topological quantum computation \cite{brennen2008}. This algorithm computes WRT invariants in $\Theta((m+\ell)N^2)$ time and $\Theta(N^2)$ space by efficiently combining the $2^m$ conceptual terms from the Frohman-Gelca product rule on an $N^2$ coefficient table, where $\ell$ is the word length of the monodromy. Second, we develop a quantum algorithm that trades space for coherence, using only $2\lceil\log_2 N \rceil + m + 1$ qubits ($O(\log N)$ data qubits plus $m+1$ ancillas)—polynomial in $\log N$ and $m$ versus the classical $\Theta(N^2)$ requirement.

Our work clarifies the complexity landscape of topological invariants through multiple perspectives. While the Jones polynomial for arbitrary links and WRT invariants for general 3-manifolds remain \#P-hard \cite{jaeger1990,alagic2017}, the special case of torus bundles admits polynomial-time computation because the non-commutative torus at roots of unity yields a fixed finite-dimensional representation. The key insight is that the $2^m$ terms in the skein product expansion, while conceptually exponential, all map to the same $N^2$ basis elements with coefficients that can be coherently updated—analogous to how $(1+x)^m$ has $2^m$ terms but can be computed efficiently in the polynomial ring.

Importantly, we identify the colinear FG-Coefficient Counting problem (Definition~\ref{def:fgcc}) as \#P-complete (Section~\ref{sec:fg-colinear}), demonstrating that computational hardness emerges when indices are treated as integers rather than elements of $\mathbb{Z}_N$. This dichotomy—polynomial-time for fixed $N$ with modular arithmetic versus \#P-hardness for the integer problem—illuminates the subtle boundary between tractable and intractable regimes in topological invariant computation. This coefficient counting problem represents a novel contribution to computational topology: it has not been previously studied from a complexity-theoretic perspective, yet it arises naturally from the algebraic structure of skein theory and provides a concrete example where exact classical counting is \#P-complete while quantum algorithms achieve polynomial-time additive approximation of the normalized coefficient---meaningful for coefficients of size $\Omega(2^m)$. Unlike some conjectured quantum advantages for coefficient problems that have been refuted by improved classical algorithms \cite{classical_quantum_multiplicities_2025}, our separation is robust—guaranteed by the \#P-completeness proof to persist unless the polynomial hierarchy collapses \cite{toda1991}.

The quantum algorithm offers compelling advantages in specific scenarios: when memory is constrained (storing $N^2$ complex numbers may be prohibitive for large $N$), when only the trace is needed (avoiding computation of the full coefficient table), or when approximate results suffice. The scaling difference—$O(\log N + m)$ total qubits compared to $\Theta(N^2)$ classical memory—represents an exponential space advantage.

Experimental implementation of our algorithms would occupy a distinct and complementary position in the landscape of quantum advantage demonstrations for topological invariants. While recent Jones polynomial experiments \cite{laakkonen2025} validate quantum advantage for knot invariants (1-dimensional objects within a 3-manifold), our algorithms target 3-manifold invariants (characterizing the 3-dimensional space itself). This dimensional distinction reflects the hierarchy within WRT theory: Jones polynomials are special cases of WRT invariants, recovered when the 3-manifold is $S^3$ with link insertions \cite{turaev1994}. Furthermore, our logarithmic qubit scaling in $N$ enables exploration of exponentially larger state spaces compared to braid-based approaches with linear qubit scaling—for instance, with 14 data qubits, our approach accesses a 16{,}129-dimensional space ($N=127$) while braid-based methods handle only 13-strand braids. Both experimental programs rest on rigorous classical hardness results: Jones polynomial evaluation is \#P-hard even for value-distinguishing approximation \cite{jaeger1990,kuperberg2015}, while our FG-Coefficient Counting is \#P-complete for exact counting (Theorem~\ref{thm:fgcc-hard}). These complementary directions—knot invariants via braids versus 3-manifold invariants via non-commutative tori, linear versus logarithmic qubit scaling—together would provide comprehensive validation of quantum computational advantage across the full spectrum of topological quantum field theory.



Several directions for future work emerge naturally. Our coefficient counting analysis focused on colinear inputs $v_i = (a_i, 0)$; for general inputs $v_i = (p_i, q_i)$, the Frohman-Gelca weight acquires quadratic phases $t^{\sum_{i<j}\varepsilon_i\varepsilon_j\det(v_i,v_j)}$, making the coefficient an Ising-like partition sum. The problem remains \#P-hard by restriction to the colinear case, but the quantum circuit naturally accumulates these phases while classical algorithms must sum $2^m$ oscillatory contributions—a setting where quantum interference may provide additional advantage. Extension to higher genus surfaces $\Sigma_{g,n}$ with $g > 1$ presents more fundamental challenges—the higher-genus skein algebras, unlike the torus, do not have defined product to sum formula making multiplication intractable. 
Additionally, the connection between our algorithms and topological quantum computing deserves further exploration, as the classical simulation of anyon braiding and our dynamic programming method are essentially equivalent. As quantum hardware continues to improve, the space-efficient quantum approach may become increasingly valuable for problems where $N$ is large enough to make classical storage prohibitive yet small enough for coherent quantum computation.

\bibliography{references}

\appendix

\input{appendix_content}

\end{document}

%% file: appendix_content.tex
\newpage
\section{Place-Free Approach and Non-Commutative Torus Embedding}
\label{app:place-free}

This appendix clarifies why our quantum algorithm naturally implements a place-free approach, simplifying circuit construction while maintaining full generality. This distinction is particularly relevant for the TQFT community familiar with place-based skein algebra representations.

\subsection{Place-Based vs.\ Place-Free Representations}

In the standard skein–Frobenius approach at a $2N$-th root of unity, one first localizes along a multiplicative set $S\subset Z\!\big(K_N(T^2)\big)$ (removing the bad locus in the character variety) to form $S^{-1}K_N(T^2)$ and $S^{-1}Z$; a \emph{place} is then a central character $\phi:S^{-1}Z\to\mathbb{C}$, and composing the normalized trace on $S^{-1}K_N(T^2)$ with $\phi$ yields a numerical pairing \cite{abdielfrohman2016}—this localization is essential because $K_N(T^2)$ is not a free $Z$–module and its fiber rank jumps (exceeding $N^2$ on the reducible locus), whereas after passing to $S^{-1}K_N(T^2)$ the rank stabilizes over $S^{-1}Z$ and the specialized trace is nondegenerate. A similar construction holds for any noncompact finite-type surface with $\chi(\Sigma)<0$ (punctures or boundary) \cite{abdielfrohman2017}.

Our place-free viewpoint canonically embeds $K_N(\Sigma_{1,0})$ into the symmetric subalgebra $\mathcal{W}_t^{\iota}$ of the noncommutative torus. We work throughout in the fixed, faithful $N^2$-dimensional left-regular model, so all algebraic manipulations and the center-valued reduced trace are independent of any place. When a scalar is required, we either compose with a central character (recovering the usual place-based value) or use a single normalization and phase convention, specified later, that aligns our normalized matrix trace with the standard WRT normalization. For reference: after fixing a central character, the noncommutative torus has a unique $N$-dimensional irreducible representation; we deliberately remain in the $N^2$-dimensional faithful model to stay place-free at the linear-algebra level.

\subsection{Advantages for Quantum Computation}

The place-free approach provides decisive advantages: \textbf{(1) Fixed circuit architecture}—every skein element requires exactly $2\lceil\log_2 N\rceil$ data qubits with uniform two-term LCU decomposition (plus one ancilla per multiplication for correct block-encoding composition), independent of topological complexity. \textbf{(2) Direct trace computation}—$\operatorname{Tr}(W)$ yields the WRT invariant directly, requiring no averaging over places, basepoint choices, or post-processing calibration. \textbf{(3) Comparison with braid-based approaches}—traditional Jones polynomial algorithms use place-based representations where Pauli decomposition widths vary with both the braid and chosen place, complicating resource estimation. Our place-free approach provides uniform, efficient computation directly from algebra to topological invariant.

\section{Mapping Tori and Modular Action Background}
\label{app:mapping}

This appendix provides essential background on torus bundles (mapping tori), particularly for readers from the quantum algorithms community. We summarize the geometric construction and explain how the modular group $\text{SL}_2(\mathbb{Z})$ acts via quantum circuits.

\subsection{Mapping Torus Construction}

A mapping torus from homeomorphism $f: \Sigma \to \Sigma$ is the quotient space $M_f := (\Sigma \times [0,1])/(x,1) \sim (f(x),0)$, which naturally fibers over $S^1$ with fiber $\Sigma$ and monodromy $f$. For the torus $\Sigma = T^2$, the mapping class group $\text{MCG}(T^2) \cong \text{SL}_2(\mathbb{Z})$ acts on homology $H_1(T^2; \mathbb{Z}) \cong \mathbb{Z}^2$, so each matrix $g \in \text{SL}_2(\mathbb{Z})$ determines a torus bundle $M_g$ over $S^1$.

\subsection{WRT Invariants and Skein Insertions}

In $(2+1)$-dimensional TQFT, the torus $T^2$ has Hilbert space $V_N(T^2)$ (level-$N$ theory), and the WRT invariant is $Z_N(M_g) = \operatorname{Tr}(\rho(g))$ where $\rho: \text{SL}_2(\mathbb{Z}) \to \text{GL}(V_N(T^2))$ is the modular representation. With $m$ skein insertions $x_1,\ldots,x_m$ along the $S^1$ fiber, the decorated invariant becomes
\[
Z_N(M_g; x_1, \ldots, x_m) = \operatorname{Tr}(\rho(g) L_{x_m} \cdots L_{x_1}),
\]
where $L_{x_i}$ is the linear operator for skein element $x_i$. This formula is the starting point for our algorithms.

\subsection{Quantum Circuit Implementation of \texorpdfstring{$\text{SL}_2(\mathbb{Z})$}{SL2(Z)}}

Standard generators $S,T$ of the modular group act on computational basis states $(p,s) \in \mathbb{Z}_N^2$ via modular arithmetic. The $T$-generator maps $(p,s) \mapsto (p+s \bmod N, s)$ with quadratic phase $t^{s^2}$, implemented via modular addition followed by a quadratic phase gate in $O(\log^2 N)$ gates. The $S$-generator maps $(p,s) \mapsto (-s \bmod N, p)$, implemented via double QFT, SWAP, and modular negation in $O(\log^2 N)$ gates.

Any $g \in \text{SL}_2(\mathbb{Z})$ decomposes as a word of length $\ell = O(\log \|g\|)$ in generators $S,T$~\cite{farb2011}, yielding quantum circuit depth $O(\ell \log^2 N)$ for $\rho(g)$. This logarithmic scaling follows from the connection between $\text{SL}_2(\mathbb{Z})$ decomposition and the Euclidean algorithm: for $g = \bigl(\begin{smallmatrix} a & b \\ c & d \end{smallmatrix}\bigr)$, the decomposition essentially runs the Euclidean algorithm on the matrix entries, which terminates in $O(\log \max(|a|,|b|,|c|,|d|))$ steps. This efficient modular action is crucial for polynomial-time quantum WRT computation.

\section{Connection to \texorpdfstring{SU(2)$_k$}{SU(2)k} Anyons and TQFT}
\label{app:technical-anyons}

This appendix provides the dictionary between the algebraic structures in our algorithms (skein algebra, non-commutative torus) and the physical framework of $SU(2)_k$ anyons and Topological Quantum Field Theory.

\subsection{\texorpdfstring{SU(2)$_k$}{SU(2)k} Anyons at Level \texorpdfstring{$k = N-2$}{k = N-2}}

The WRT invariants computed in this work correspond to $SU(2)$ Chern-Simons theory at level $k = N-2$. In the language of Modular Tensor Categories, the relevant category $\mathcal{C} = SU(2)_k$ is pre-modular, becoming modular upon restricting to integrable representations~\cite{turaev1994, reshetikhin1991}.

The simple objects (anyon types) are labeled by half-integers $j \in \{0, 1/2, 1, \ldots, k/2\}$, corresponding to integrable highest weight representations of $\widehat{\mathfrak{sl}}(2)_k$. Using integer labels $a = 2j \in \{0, 1, \ldots, N-2\}$, the label $a = 0$ corresponds to the vacuum. The fusion of anyons $a$ and $b$ follows the truncated Clebsch-Gordan series~\cite{turaev1994}:
\[
a \times b = \sum_{c=|a-b|, \text{ step } 2}^{\min(a+b, 2k-(a+b))} c,
\]
with quantum dimension $d_a = [a+1]_q = \sin(\pi(a+1)/N)/\sin(\pi/N)$. The modular $S$-matrix governing braiding statistics has entries $S_{ab} = \sqrt{2/N} \sin(\pi(a+1)(b+1)/N)$, which diagonalizes the fusion rules via the Verlinde formula~\cite{reshetikhin1991}.

\subsection{Physical Interpretation of Trace Components}

The trace formula $Z_N(M_g; x_1, \ldots, x_m) = \operatorname{Tr}(\rho(g) L_{x_m} \cdots L_{x_1})$ admits direct physical interpretation in terms of TQFT spacetime histories~\cite{witten1989, atiyah1988}. The operator $L_{B_{(p,s)}}$ represents insertion of a Wilson loop carrying the fundamental representation along the curve class $(p, s) \in H_1(T^2; \mathbb{Z})$; in the 3-manifold $M_g$, this corresponds to a colored link component winding around the fiber. The algebraic relation $L_{B_{(p,s)}} = U_{p,s} + U_{-p,-s}$ reflects decomposition into flux operators in the non-commutative torus.

The modular actions $\rho(S)$ and $\rho(T)$ represent large diffeomorphisms on the spatial torus: $\rho(S)$ implements an $S$-transformation (90-degree rotation of meridian and longitude) generated by the modular $S$-matrix, while $\rho(T)$ implements a Dehn twist imparting phase $e^{2\pi i h_a}$ to anyon state $|a\rangle$ with topological spin $h_a = a(a+2)/4N$. The trace operation corresponds to gluing the two boundaries of $T^2 \times [0,1]$ to form $M_g$, summing over intermediate anyonic states to yield the partition function~\cite{atiyah1988}.

\subsection{Classical Simulation and Quantum Advantages}

The classical algorithm (Theorem~\ref{thm:complexity}) ``simulates'' these anyons by computing exact values of topological invariants---partition functions and Wilson loop expectation values---that would be measured in an ideal topological quantum computer. By maintaining the coefficient table in the non-commutative torus basis, the algorithm explicitly tracks destructive and constructive interference of topological paths (the $2^m$ Kauffman bracket terms).

While the classical algorithm runs in polynomial time for torus bundles, the quantum algorithm (Theorem~\ref{thm:total-complexity}) provides intrinsic advantages. The classical simulation requires $\Theta(N^2)$ memory for the full TQFT state vector, whereas the quantum algorithm stores this coherently in $O(\log N)$ qubits---for large $N$, this maps an intractable classical memory requirement to a small qubit register. Additionally, extracting specific coefficients from the anyonic interference pattern is \#P-complete classically (Theorem~\ref{thm:fgcc-hard}), while the quantum algorithm estimates these efficiently via the Hadamard test.